\documentclass[12pt, draftclsnofoot, onecolumn]{IEEEtran}
\usepackage{amsthm}  
\usepackage{amsfonts}
\usepackage{color}
\usepackage[normalem]{ulem}
\usepackage{graphicx}
\usepackage[tight,footnotesize]{subfigure}
\usepackage{psfrag}
\usepackage{subfigure}
\usepackage{epsfig}
\usepackage{amssymb}
\usepackage{tabu}
\usepackage{cite} 

\usepackage{algorithmic,algorithm}
\newtheorem{definition}{Definition}

\makeatother
\usepackage{multirow}
\usepackage[cmex10]{amsmath}
\usepackage{bm}

\newcounter{rcounter}
\allowdisplaybreaks[4]

\newtheorem{remark}{Remark}
\newtheorem{theorem}{Theorem}

\newtheorem{lemma}{Lemma}

\newtheorem{corollary}{Corollary}

\newtheorem{proposition}{Proposition}

\def\ScaleIfNeeded{%
\ifdim\Gin@nat@width>\linewidth \linewidth \else \Gin@nat@width
\fi } \makeatother

\begin{document}
%
\title{ Optimal User Scheduling and Power Allocation for Millimeter Wave NOMA Systems}

\vspace{-2.9em}
\author{
Jingjing Cui,~\IEEEmembership{Student Member, IEEE,} 
Yuanwei~Liu,~\IEEEmembership{Member,~IEEE,}\\
Zhiguo Ding,~\IEEEmembership{Senior Member, IEEE,} 
Pingzhi Fan,~\IEEEmembership{Fellow, IEEE} and  
Arumugam~Nallanathan,~\IEEEmembership{Fellow,~IEEE,}

\thanks{J. Cui and P. Fan are  with the Institute of Mobile Communications, Southwest Jiaotong University, Chengdu 610031, P. R. China. (email: cuijingj@foxmail.com, p.fan@ieee.org).}

\thanks{Y. Liu and A. Nallanathan are with the Department of Informatics, King's College London, London WC2R 2LS, U.K. (email: \{yuanwei.liu, arumugam.nallanathan\}@kcl.ac.uk).}

\thanks{ Z. Ding is  with the School of Computing and Communications, Lancaster University, LA1 4YW, UK. (e-mail:
z.ding@lancaster.ac.uk). }


\thanks{Part of this work submitted in IEEE Global Communication Conference (GLOBECOM),  Dec. Singapo, 2017 \cite{17GlbCom}.}

\thanks{}  }

\maketitle
\vspace{-2.0em}
\begin{abstract}
This paper investigates the application of non-orthogonal multiple access (NOMA) in millimeter wave (mmWave) communications by exploiting beamforming, user scheduling and power allocation. 
Random beamforming  is invoked for reducing the  feedback overhead of considered systems.  A  non-convex optimization problem for maximizing the  sum rate is formulated, which  is proved to be NP-hard.   The branch and bound (BB) approach is invoked to obtain the $\epsilon$-optimal power allocation policy, which is proved to converge to a  global optimal  solution. 
To elaborate further, a low complexity suboptimal approach is developed for striking a good  computational complexity-optimality tradeoff, where matching theory and successive
convex approximation (SCA) techniques  are  invoked for tackling the  user scheduling and power allocation problems, respectively.
 Simulation results reveal that: i) the proposed low complexity   solution achieves a near-optimal performance; and ii)  the proposed  mmWave NOMA systems is capable of  outperforming  conventional mmWave orthogonal multiple access (OMA) systems in terms of sum rate and the number of served users.
\end{abstract}
\vspace{-0.5em}
\begin{IEEEkeywords}
\vspace{-0.5em}
Millimeter wave (mmWave), non-orthogonal multiple access (NOMA), power allocation, user scheduling
\end{IEEEkeywords}

\vspace{-0.5em}
\section{Introduction}
\vspace{-0.3em}
The unprecedented demand for high data rates imposes challenges for fifth generation (5G) networks.
Millimeter wave (mmWave) communication has been viewed as a promising candidate technology to address the challenge of bandwidth shortage \cite{Pi11Mag}, due to the large bandwidths in the mmWave spectrum. Particularly, advances in mmWave hardware  and the potential availability of spectrum have encouraged the wireless industry to
consider mmWave for the access link in outdoor cellular systems. Different from the propagation characteristics in the sub-6GHz wireless communication, the propagation in the mmWave band is highly directional with severe propagation path loss, low penetration coefficients and high signal attenuation \cite{Deng15ICCW,Rappaport12ICC}. 
 To compensate the large path loss in the mmWave band,  directional beamforming provides an effective solution to resist the large path loss as well as to provide sufficient received signal power  \cite{Alkhateeb14Mag}.

Non-orthogonal multiple access (NOMA) in power domain provides a superior spectral efficiency and hence has recently received significant attentions  \cite{Saito13VTC}.  The key idea of NOMA is to multiplex multiple users on different power levels for multiple access within a given resource block (e.g., time/frequency). Moreover, it particularly invokes  successive interference cancellation (SIC) techniques at receivers who have better channel conditions for removing intra-channel  interference. As a result, NOMA is capable of supporting massive connectivity and efficiently meeting the users' diverse quality of service (QoS) requirements \cite{Ding17Mag}. 

Sparked by the aforementioned characteratics of mmWave communication and NOMA, the use of NOMA in mmWave sprectrum is highly desired due to the following advantages:
 \begin{itemize}
  \item The highly directional transmission in mmWave spretrum implies that the users' channels  can be severely correlated, which are suitable for applying NOMA technique.  
  \item Directional beams in mmWave communication with large-scale arrays bring large antenna array gains and small inter-beam interference, enabling NOMA transmission over each beam.
  \item Applying NOMA into mmWave communication is capable of enhancing the spectral efficiency, which provides a  new highly rewarding candidate for 5G networks.
 \end{itemize}



\vspace{-0.3cm}
\subsection{Prior Works}
\vspace{-0.3cm}

\subsubsection{Studies on mmWave systems} 
In constrast to the conventional low frequency multiple-input multiple-output (MIMO) systems, the additional radio frequency (RF) hardware constraints such as high-resolution analog-to-digital converters (ADCs) exist in mmWave systems. Hence, fully digital baseband beamforming becomes impossible \cite{Ayach14TWC}.  Considering the high power consumption of mixed signal components in mmWave system,  the hybrid analog and digital beamforming was proposed in \cite{Hur13TC}. Since analog beamforming is implemented via analog phase shifters, the modulus of the elements in the analog beamforming vectors are constrained to a constant. The hybrid analog and digital beamforming for  mmWave systems was studied in 
\cite{Ayach14TWC,Hur13TC,Yu16JSTSP,Sohrabi16JSAC}, where the designs of the beamforming matrices are in general based on perfect channel state information (CSI).  
Unfortunately, in practice, accurate channel estimation and CSI feedback to the base station (BS) are difficult \cite{Alkhateeb15TWC,Lee16JSTSP}, which induce heavy system overhead particularly in  multi-user mmWave  downlink systems.  To reduce the feedback overhead, a two-stage hybrid analog and digital beamforming approach was proposed in  \cite{Alkhateeb15TWC}, where the analog beamforming designs at the BS and the users are constructed for maximizing its own desirable signal based on individual CSI. 
 In addition, random beamforming provides an effective approach in reducing the CSI feedback \cite{Tse05wireless}.  The performance of random bemforming in conventional mmWave systems was investigated in \cite{Lee16JSTSP}. It was shown that random beamfoming in mmWave channels is indeed capable of achieving a very good sum rate performance with the aid of user scheduling and power allocation strategies. 


\subsubsection{Studies on single-input single-output (SISO)-NOMA systems}
Early research contributions have studied the potential implementation of NOMA in cognitive radio (CR) networks \cite{Ding16TVT,Liu17TWC}  and simultaneous wireless information and power transfer (SWIPT) protocol  \cite{Liu16JSAC}. More particularly, in \cite{Ding16TVT}, the impact of user paring on the transmission sum rate was investigated both for fixed power allocation NOMA systems and CR inspired NOMA.  As the interplay between NOMA and CR is bi-directional, the application of NOMA in large-scale CR networks was exploited in \cite{Liu17TWC} with carefully considering the channel ordering issue.  Aiming at addressing energy related issues, in \cite{Liu16JSAC}, a novel comparative NOMA scheme was proposed by  invoking simultaneous wireless information and power transfer (SWIPT) technique. 
Regarding the resource allocation works in NOMA, a joint subcarrier and power allocation algorithm was developed in \cite{LeiYHS16}, where a near optimal solution was developed based on Lagrangian duality and dynamic programming. In \cite{Fang16ITC}, a low-complexity suboptimal algorithm based on matching theory was developed to maximize energy efficy for multi-subcarrier (MC)-NOMA systems.  The authors of \cite{DiTWC16}   
proposed  an asymptotically optimal joint power allocation and user scheduling algorithm based on matching theory to maximize the   sum rate of  MC-NOMA systems.  Furthermore, in \cite{Sun17TC}, an effective power allocation and user scheduling algorithm based on monotonic optimization theory was proposed for full-duplex MC-NOMA systems. Driven by the partial CSI feedback, a power allocation strategy for SISO-NOMA systems based on the average CSI was developed in \cite{Cui16SPL}. 

\subsubsection{Studies on multiple-input multiple-output (MIMO)-NOMA  systems}
   In \cite{choi2015minimum}, the author proposed a beamforming design approach to minimize transmission power  where a multi-antenna base station  (BS) communicates two single-antenna NOMA users in each beam. In \cite{Hanif16TSP}, a multi-antenna BS performs NOMA transmission with $K$ single-antenna users via designing the beamforming vectors, in which an effective channel gain constraint was formulated to guarantee users' fairness. Based on these studies,  the authors of \cite{Ding16TWC} proposed an general MIMO-NOMA designing framework, where  users were firstly grouped into small-size clusters, and then the NOMA principle was employed for each cluster.   
Furthermore, in  \cite{Liu16CL}, a user clustering and power allocation scheme was proposed to optimize the user fairness of MIMO-NOMA systems. 
be two. 




\vspace{-0.3cm}
 \subsection{Motivation and Contributions}
 \vspace{-0.3cm}
While the aforementioned research contributions have laid a solid foundation on mmWave and  NOMA systems, the investigations on the applications of NOMA on mmWave band are still in their fancy. Moreover, whether NOMA technique is capable of bringing perfromance enhancement for mmWave networks are still unknown. 
In this paper, we study the mmWave NOMA system, where the BS generates some separable beams and then NOMA transmission is applied on each beam.  
It is worth pointing out that the characteristics  of mmWave propagation  makes it impossible that applies the  digital beamforming which was invoked in the conventional sub-6GHz MIMO-NOMA works. 
In order to reduce the feedback overhead, the work of \cite{Ding17Ac} studied the co-existence of NOMA and mmWave systems with random beamforming, which  showed that the performance of the mmWave NOMA systems  outperforms conventional mmWave-OMA systems.  
The advantage of random beamforming applied in  mmWave NOMA systems is that only equivalent channel gains of all users are required at the BS. In an effort for improving the performance of random beamforming, an efficient user scheduling method is required. 
 Moreover, power allocation strategies among inter/intra-beams are capable of further enhancing the performance of  mmWave-NOMA systems. In addition,the existence of inter-beam interference  in mmWave NOMA systems,  makes the  user scheduling and the power allocation become  more challenging and fundamentally different from the existed works for MC-NOMA systems \cite{LeiYHS16,DiTWC16,Sun17TC}. 

Driven by solving all the aforementioned issues, in this paper, we investigate the mmWave systems with adopting NOMA techniques under partial CSI feedback. 
More specifically, the BS first generates a set of random beams,  then each user feedback its scale channel gain to the BS, which avoids the cumbersome system overhead on the feedback of channel vectors. By doing so, the idealized perfect CSI assumption adopted in aforementioned MIMO-NOMA works [22-25] are relaxed.   Based on these channel information, the BS 
schedules multiple users  on each predefined beam, and then transmits the superposed signals based on NOMA with allocating appropriate power for each beam as well as users.  To the best of our knowledge, this is the first work to jointly consider  user scheduling and power allocation strategies in mmWave NOMA systems.
  Our main contributions are summarized as follows.

\begin{enumerate}
\item We propose a general downlink mmWave NOMA systems with the aid of random beamforming, in which the BS  requires the scale channel gains of all users rather than to obtain all channel vectors of users.   Then, we formulate  the  sum rate  maximization problem  subject to the users' QoS requirements by designing  the user scheduling and power allocation strategy. We mathematically proved that the formulated problem is non-deterministic polynomial-time (NP)-hard.

\item We decompose the original non-convex problem into two  subproblems as  \emph{user scheduling} and  \emph{power allocation}. By leveraging the branch and bound (BB) approach, we propose a global optimal solution for the power allocation.


\item  We  develop a low complexity solution with the aid of \emph{matching theory} and \emph{successive convex approximation (SCA)}.
Firstly, based on the concept of stable matching, we propose a low complexity suboptimal algorithm.  Secondly, we propose an efficient SCA  algorithm for providing a high-quality power allocation solutions. The properties  of the matching and SCA algorithms are analyzed.

\item  We  demonstrate that the proposed mmWave NOMA framework outperforms the conventional mmWave OMA framework with the aid of both of the proposed algorithms. 
Moreover, the proposed low complexity solution are capable of  achieving a near-optimal performance.

\end{enumerate}
  
 \vspace{-0.3cm}
\subsection{Organization}
\vspace{-0.3cm}
The rest of the paper is organized as follows. In Section II, the system model for studying mmWave NOMA and the random beamforing scheme are presented. The joint user scheduling and power allocation problem are formulated in Section III.  In Section IV, a global optimal solution based on BB is provided and a low complexity  power allocation and user scheduling  algorithm are developed in Section V. Simulation results are presented in Section VI, which is followed by conclusions in Section VII. 
 
 
 \vspace{-0.5cm} 
 \section{System model}
 \vspace{-0.3cm}
\subsection{Signal Model}\label{SigMod}
\vspace{-0.3cm}
Consider an mmWave-NOMA downlink scenario composed of  one BS  with  $N_{\mathrm{RF}}$ transmit antennas and $K$  single antenna users.  Assuming that the BS performs MIMO transmission with $M$ beams, $K\geq 2M$.  Denote by $\mathcal{M}=\{1,\cdots,M\}$ and $\mathcal{K} = \{ 1, \cdots,K\}$ be the beam set and the user set, respectively.
The $m$-th transmit beamforming vector is denoted as $\mathbf{w}_m \in \mathcal{C}^{M \times 1}$. 
  We assume that the multiuser scheduler schedules $q_m$ users denoted by $\mathcal{C}_m$ on  the $m$-th beam to perform NOMA and $\mathcal{C} =  \bigcup_{m\in \mathcal{M}}\mathcal{C}_m$ is the set of the total scheduled  users. We further assume that each user is scheduled by a single beam at most; thus, $\mathcal{C}_m \bigcap \mathcal{C}_n =\emptyset$, $n \neq m$. Let $ c_{k}^m$ indicate the indicators for user $k$ on  the $m$-th beam, $c_{k}^m \in \{0,1\}$. If $c_{k}^m =1$,  it indicates user $k$ is  scheduled on beam $m$ and $c_{k}^m =0$ if otherwise.  Let $s_{k}$ denote the data symbol transmitted for user $k$ and $\beta_{k}^m$  be the transmission power assigned for user $k$ on the $m$-th beam. We define  $M_t = |\mathcal{C}|$ to denote the total number of the scheduled  users. The total transmission power satisfies $ \sum_{m=1}^M \sum_{k=1}^{K}c^m_k\beta_{k}^m \leq P_{tot}$, where $P_{tot}$ is the maximum transmission power of the BS.
\begin{figure} [t!]
\centering
\includegraphics[width= 3.5in, height=2.2in]{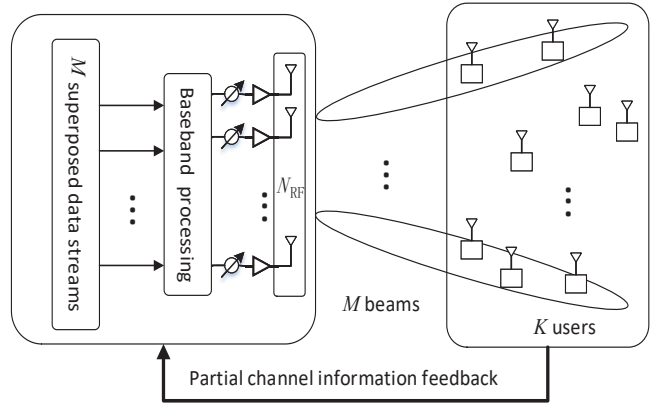}
\vspace*{-1.5em} \caption{System model for mmWave-NOMA transmission in downlink MISO scenario. }
 \label{system_model}
 \vspace{-1.5em}
\end{figure}

In the proposed mmWave-NOMA system, the BS chooses $M_t$ users among the $K$ users in the cell and broadcasts $M$ independent superposed data streams to the $M_t$ selected users with beamforming matrix $\mathbf{W}=\{\mathbf{w}_1,\cdots,\mathbf{w}_M\}$. 
Assuming user $k$ is scheduled at the $m$-th beam, the received signal at user $k$ is 
{\setlength\abovedisplayskip{3pt} 
\setlength\belowdisplayskip{3pt}
{\small
\begin{eqnarray}
\begin{aligned}
y_{k}^m = &\begin{matrix} \underbrace{\mathbf{h}_{k}^{mH}\mathbf{w}_{m}c_k^m\sqrt{\beta_{k}^m} s_{k}}\\  \text{ \small{\bf Desired signal}}\end{matrix}+ \begin{matrix} \underbrace{\mathbf{h}_{k}^{mH}\mathbf{w}_{m}\begin{matrix} \sum_{j\neq k} c_j^m\sqrt{\beta_{j}^m} s_{j} \end{matrix}  } \\ \small{\text{ \bf Intra-beam interference}} \end{matrix}   + \begin{matrix}\underbrace{ \begin{matrix} \sum_{n \neq m} \sum_{i\in \mathcal{K}}  \mathbf{h}_{k}^{mH} \mathbf{w}_{n} c_i^n \sqrt{\beta_{i}^n} s_{i} \end{matrix} }\\\small{\text{ \bf Inter-beam interference}} \end{matrix} + \begin{matrix} \underbrace{\nu_{k}}\\ \small{\text{ \bf Noise}} \end{matrix},
\end{aligned} 
\end{eqnarray}}}with $j,k \in \mathcal{K}$ and $m,n\in \mathcal{M}$, where $\mathbf{h}_{k}^m \in \mathbb{C}^{N_{\mathrm{RF}} \times 1}$ be the mmWave channel between the BS and  user $k$ and  $\nu_k \sim \mathcal{CN} (0,\sigma^2)$ is the additive white Gaussian noise for user $k$. It is assumed that all users have the same noise power in this paper.

\subsection{Channel Model}
Different from the conventional low frequency channel, the mmWave channel  in general has limited scattering due to the high free-space path loss. Thus, we consider the geometric channel model which can embody the low rank and spatial correlation characteristics of mmWave communications \cite{Ayach14TWC,Alkhateeb15TWC}. Using this model the channel from the BS to user $k$ can be modelled as
{\setlength\abovedisplayskip{3pt} 
\setlength\belowdisplayskip{3pt}
{\small
\begin{align}\label{Chmodl.eq}
\mathbf{h}_{k} =\sqrt{M \rho_k}\sum_{l=1}^L a_{k,l}\mathbf{a}_{BS}(\theta_{k,l}) ,
\end{align}}}where $\rho_k$ denotes the average path-loss between the BS and user $k$. In a mmWave propagation model, $\rho_k$ is  given by $\rho_k = \eta d_k^{-\alpha}$, where $\eta = \big(\frac{c}{4 \pi f_c}\big)^2$ is the frequency independent constant with $c= 3\times 10^8 m/s$ and the carrier frequency $f_c$. Thus, the valuses of $\eta$ are different for different mmWave frequencies.   $d_k$ is the distance between the BS and user $k$ and $\alpha$ is the path loss exponent depending on the line-of-sight (LoS) and non-line-of-sight (NLoS) links, i.e., $\alpha = \alpha_{LoS}$ for LoS link and $\alpha = \alpha_{NLoS}$ for NLoS link. In this paper, we assume that $l=1$ is the LoS link. Furthermore, $a_{k,l}$ is the complex gain of the $l$-th path with $a_{k,l} \sim \mathcal{CN}(0,1)$. $\theta_{k,l}$ denotes the $l$-th path's normalized direction to the physical angle of departure $\phi_{k,l} $ with $\phi_{k,l} \in [0,2\pi]$ and $\theta_{k,l} = \frac{2d\sin(\phi_{k,l})}{\lambda}$, where $\lambda$ is the signal wavelength, and $d$ is the distance between antenna elements. At last, $\mathbf{a}_{BS}^H(\theta_{k,l})$ is the antenna array response vectors of the BS. In this paper, we consider a uniform linear array (ULA), where $\mathbf{a}_{BS}^H(\theta_{k,l})$ can be defined as
{\setlength\abovedisplayskip{3pt} 
\setlength\belowdisplayskip{3pt}
{\small
\begin{align}\label{ArrayResCh.eq}
\mathbf{a}_{BS}(\phi_{k,l}) = \frac{1}{\sqrt{M}}\left[1,e^{j\pi \theta_{k,l}},\cdots, e^{j(M-1)\pi \theta_{k,l}} \right]^T.
\end{align}}}
\noindent

\vspace{-1cm}
\subsection{Analog Beamforming}

Due to the high cost and power consumption for hardware constraints,  a low-complexity analog beamforming is adopted in this paper. Specifically, we consider the random beambeamforming scheme to reduce the feedback overhead, where the direction of each analog beamforming vector is predefined. Suppose that the BS will form $M$ orthogonal beams for NOMA transmission. These beams  are predefined and are known to the BS and the users prior to transmission. Following  \cite{Lee16JSTSP}, these $M$ orthogonal beamforming vectors can be constructed by
{\setlength\abovedisplayskip{3pt} 
\setlength\belowdisplayskip{3pt}
{\small
\begin{align}
\mathbf{w}_m = \mathbf{a}\Big( \zeta + \frac{2(m-1)}{M} \Big),
\end{align}}}where $\zeta$ denotes a random variable following a uniform distribution with $\zeta \in [-1, 1]$. 

Assuming each user computes $M$ equivalent channel gain and feedbacks the magnitudes $\{g_{k}^m = |\mathbf{h}_{k}^H\mathbf{w}_m|^2,~m\in \mathcal{M}\}$ and the corresponding beam indices to the BS. With this information, the BS performs user scheduling and power allocation, which will be discussed in the following sections.

 \vspace{-0.5cm} 
\section{Problem formulation}
 \vspace{-0.3cm} 

Since multiple users are supported on each beam, based on the principle of NOMA, each user tries to employ SIC in a successive order to remove the intra-beam interference. Hence the decoding order is an essential issue for the  mmWave-NOMA systems. Let $\pi_m(k)$ be the decoding order for user $k$ on beam $m$,  namely, if $\pi_m(k) = i$, then user $k$  scheduled on beam $m$ is the $i$-th signal to be decoded. 
For any two users $j$ and $k$ scheduled on beam $m$ satisfying $\pi_m(j) \leq \pi_m(k)$, the signal-to-interference-
plus-noise ratio (SINR) of user $k$  to decode user $j$ is given by
{\setlength\abovedisplayskip{3pt} 
\setlength\belowdisplayskip{3pt}
{\small
\begin{eqnarray}
\begin{aligned}
\mathrm{SINR}_{j\rightarrow k}^m=\frac{c_j^m g_{k}^m \beta_{j}^m }{ g_{k}^m \sum\limits_{\pi_m(i)>\pi_m(j)} c_i^m \beta_{i}^m +\sum\limits_{n \neq m} g_{k}^n \beta^n + \sigma^2},
\end{aligned} 
\end{eqnarray}}}with $i,j,k \in \mathcal{C}_m$ and $m \in \mathcal{M}$, where $\beta^n = \sum_{i=1}^K c_i^n\beta_i^n$ is the transmission power on beam $n$. The corresponding decoding rate is $R_{j \rightarrow k}^m = \log_2(1+\mathrm{SINR}_{j\rightarrow k}^m)$. 
The achievable SINR for user $j$ on beam $m$  can be expressed as
{\setlength\abovedisplayskip{3pt} 
\setlength\belowdisplayskip{3pt}
{\small
\begin{eqnarray}\label{Rjj.eq}
\begin{aligned}
\mathrm{SINR}_{j\rightarrow j}^m = \frac{c_j^m g_{j}^m \beta_{j}^m }{ g_{j}^m \sum\limits_{\pi_m(i)>\pi_m(j)} c_i^m\beta_{i}^m +\sum\limits_{n \neq m} g_{j}^n \beta^n + \sigma^2},
\end{aligned} 
\end{eqnarray}}}with  $i,j \in \mathcal{C}_m$ and $m \in \mathcal{M}$. The correponding rate is $R_{j\rightarrow j }^m = \log_2(1 + \mathrm{SINR}_{j\rightarrow j}^m)$. Under the assumption of a given decoding order, to guarantee SIC performed successfully, the condition $R_{j\rightarrow k}^m \geq R_{j\rightarrow j}^m$ for $\pi_m(k)\geq \pi_m(j), ~j,k\in \mathcal{C}_m$ should be kept. For example, we assume that two users on beam $m$. Given the decoding $\pi_m(j)=j,~j=1,2$,   the SIC decoding condition at user $2$  can be expressed as
{\setlength\abovedisplayskip{3pt} 
\setlength\belowdisplayskip{3pt}
{\small
\begin{subequations}
\begin{align}
R_{1 \rightarrow 2}^m  \geq  R_{1 \rightarrow 1}^m, 
\end{align}
\end{subequations}}}\par When three users are allowed on beam $m$, the SIC decoding condition at user $2$  and user $3$ under decoding order $\pi_m(j)=j,~j=1,2,3$ can be given by
{\setlength\abovedisplayskip{3pt} 
\setlength\belowdisplayskip{3pt}
{\small
\begin{subequations}
\begin{align}
\begin{cases}
R_{1 \rightarrow 2}^m  \geq  R_{1 \rightarrow 1}^m, \\
R_{1 \rightarrow 3}^m  \geq  R_{1 \rightarrow 1}^m,\\
R_{2 \rightarrow 3}^m  \geq  R_{2 \rightarrow 2}^m.
\end{cases}
\end{align}
\end{subequations}}}\par It is easy to know that there will be $2^{q_m-1}-1=\sum_{k=1}^{q_m-1}\tbinom{k}{2}$ constraints when $q_m$ users are multiplexed on a single beam. 

The goal of the paper is to maximize the sum rate  subject to the total power constraint, the QoS constraints for each scheduled user and the optimal  decoding order by scheduling $M_t$ users from the $K$ users.   It can be formulated as follows.
{\setlength\abovedisplayskip{3pt} 
\setlength\belowdisplayskip{3pt}
{\small
\begin{subequations}\label{opt1.eq}
\begin{align}
 \max_{c,\beta,\pi} \quad &  \sum_{m=1}^M\sum_{j=1}^{K} R_{j \rightarrow j}^m \\
\mathrm{s.t.}\quad   &R_{j \rightarrow k}^m \geq R_{j \rightarrow j}^m,~\pi_m(k)>\pi_m(j),\label{C1.opt1}\\
& \sum_{m=1}^M \sum_{j\in \mathcal{C}_m} \beta_{j}^m \leq P_{tot},\label{C2.opt1}\\
&R_{j \rightarrow j}^m \geq \bar{R}_{j}, \label{C3.opt1} \\
& \sum_{k=1}^K c_{k}^m = q_m,\label{C4.opt1}\\
&\sum_{m=1}^M c_{k}^m \leq 1\label{C5.opt1}\\
&\pi_m \in \Pi,~j, k \in \mathcal{K},~m \in \mathcal{M}.\label{C6.opt1}
\end{align}
\end{subequations}}}where $c = \{c_k^m| k\in \mathcal{K},m\in \mathcal{M}\},\beta = \{\beta_k^m| k\in \mathcal{K},m\in \mathcal{M} \}$ and $\pi = \{\pi_m(k), k\in \mathcal{K},m\in \mathcal{M}\}$ denote the optimization variable sets of the  users, the power allocation coefficients and the decoding order, respectively.  Furthermore, $\Pi$ denotes the set of all possible SIC decoding orders.  Constraint \eqref{C1.opt1} guarantees the optimal decoding order which ensure that the SIC can be performed successfully \cite{Ding16TVT} and constraint \eqref{C2.opt1} is the total tansmission power constraint. Constraint  \eqref{C3.opt1} guarantees the QoS for  user $\pi_m(j)$ \cite{Choi16TWC}. Due to the constraint on the detection complexity of SIC receiver, we assume that   each beam can be shared by $q_m$   users, $q_m \geq 2$, in constraint \eqref{C4.opt1}.  Constraint \eqref{C5.opt1} indicates that each user can occupy one beam at maximum.

\begin{theorem} \label{Pro1}
Problem \eqref{opt1.eq} is a NP hard problem. More specifically,  problem \eqref{opt1.eq} is NP hard even only consider the  power allocation or user scheduling. 
\end{theorem}
\begin{proof}
See Appendix A.
\end{proof}

Since problem \eqref{opt1.eq} is NP hard, which results in solving problem \eqref{opt1.eq} directly becomes intractable. In the following sections, we will develop a optimal solutions based on BB techniques; then, a low complexity algorithm based on matching theory and SCA techinique will be proposed by exploiting the properties of the optimization problem itself.




 \vspace{-0.5cm} 
\section{Global optimal solutions}
 \vspace{-0.3cm} 
In this section, we try to solve problem  \eqref{opt1.eq} optimally to obtain a  global solution as a baseline. However, optimization problem \eqref{opt1.eq} contains three multi-dimensional variables: two combinational variables-$c$ and $\pi$ and one continuous variable-$\beta$.  Considering the user scheduling and the decoding order are combinational integer variables, exhaustive search is a straightforward and basic method to find the optimal solution of  integer programming problems \cite{horst2013global}.
 Then for given the scheduled users and the corresponding decoding order, we develop an   optimal power allocation  strategy based on BB techniques \cite{boyd2004convex} in the following.

Specifically, for given a set of $c$ and $\pi$, the sum rate maximization  problem in \eqref{opt1.eq} can be simplifies as follows. For notation simplicity, let $j_m$ denote the $j$-th  decoded user index scheduled on beam $m$ in the following.  
{\setlength\abovedisplayskip{3pt} 
\setlength\belowdisplayskip{3pt} 
{\small
\begin{subequations}\label{optbb1.eq}
\begin{align}
\max_{\beta} \quad &  \sum_{m=1}^M\sum_{j_m=1}^{q_m}  R_{j_m\rightarrow j_m}^m \\
\mathrm{s.t.}\quad  & R_{j_m \rightarrow k_m}^m \geq R_{j_m \rightarrow j_m}^m,\label{C1.optbb1}\\
& \sum_{m=1}^M \sum_{j_m=1}^{q_m} \beta_{j_m}^m \leq P_{tot},\label{C2.optbb1}\\
&R_{j_m \rightarrow j_m}^m \geq \bar{R}_{j_m}, \label{C3.optbb1}\\
&k_m>j_m,~j_m, k_m \in \mathcal{C}_m,~m \in \mathcal{M},\label{C4.optbb1}
\end{align}
\end{subequations}}}which is a subproblem of the original optimization problem in \eqref{opt1.eq}, since the optimization of problem \eqref{optbb1.eq} only relates with the power allocation coefficients.
Note that the objective and the constraint \eqref{C3.optbb1} contains a difference of concave functions in $\beta$. These features make problem \eqref{optbb1.eq} is still NP-hard based on {\bf Theorem \ref{Pro1}}. 

Due to the total transmission power constraint and the QoS constraints for the scheduled users, problem \eqref{optbb1.eq} may not be always feasible for example when the channel condition of the scheduled user is extremely  poor, its QoS can not be guaranteed even to be allocated by the total  power. Before solve problem \eqref{optbb1.eq}, we check the feasibility of problem \eqref{optbb1.eq} first. 

\begin{proposition} \label{Pro:Fea}
The feasibility of optimization problem \eqref{optbb1.eq} can be  checked by solving the following convex problem:
{\setlength\abovedisplayskip{3pt} 
\setlength\belowdisplayskip{3pt}
{\small
\begin{eqnarray}\label{optbb0.eq}
\begin{aligned}
 P' = \min_{\beta}\quad &\sum_{m=1}^M \sum_{j_m\in \mathcal{C}_m} \beta_{j_m}^m \\
\mathrm{s.t.}\quad  & \eqref{C1.optbb1}  \And \eqref{C3.optbb1}\And \eqref{C4.optbb1}.
\end{aligned}
\end{eqnarray}}} 
\end{proposition}
The detailed proof of {\bf Proposition \ref{Pro:Fea}} can be referred as \cite{Zheng09ITC}. Note that problem \eqref{optbb0.eq} is a power minimization problem,  which can be solved directly via a standard optimization tool. 

For given $\mathbf{c}$ and $\pi$, if optimization  problem \eqref{optbb0.eq}
is infeasible or the optimal objective value $ P'> P_{tot}$, then the given $\mathbf{c}$ and $\pi$ can not be optimal. It implies we can not find a set of feasible power allocation coefficients under given $\mathbf{c}$ and $\pi$ satisfying the optimal decoding order. Hence,   the given $\mathbf{c}$ and $\pi$ is can not be optimal.

 \vspace{-0.5cm} 
\subsection{Problem Transformations for BB Algorithms}
 \vspace{-0.3cm} 
Though the  sum rate maximization problem in \eqref{optbb1.eq} is nonconvex, it is possible to find a optimal solution based on a BB technique \cite{Weeraddana11TSP}. 
The basic idea using BB  is to optimize the objective function over a multi-dimensional rectangle.

To elaborate further, we introduce a variable set $\{\Gamma_{j_m\rightarrow j_m}^m, ~j_m\in  \mathcal{C}_m\}$, where $\Gamma_{j_m \rightarrow j_m}^m$ denotes the achievable SINR for user $j_m$ under given decoding order $\pi$. Similar to \eqref{Rjj.eq}, it can be written as
{\setlength\abovedisplayskip{3pt} 
\setlength\belowdisplayskip{3pt}
{\small
\begin{align}\label{SINR.eq1}
\Gamma_{j_m \rightarrow j_m}^m =  \frac{g_{j_m}^m \beta_{j_m}^m }{ g_{j_m}^m \sum_{i_m=j_m+1}^{q_m} \beta_{i_m}^m +\sum_{n \neq m} g_{j_m}^n \beta^n + \sigma^2},
\end{align}}}with $i_m,j_m\in \mathcal{C}_m,~m \in \mathcal{M}$.

Furthermore,  to give some useful sights, we rearrange constraint \eqref{C1.optbb1} as 
{\setlength\abovedisplayskip{3pt} 
\setlength\belowdisplayskip{3pt}
{\small
\begin{eqnarray}\label{SICOrd2.eq}
\begin{aligned}
\sum_{n\neq m} \Big(g_{k_m}^m  g_{j_m}^n - g_{j_m}^m g_{k_m}^n\Big)\beta^n + \big(g_{k_m}^m - g_{j_m}^m \big) \sigma^2 \geq 0,
\end{aligned}
\end{eqnarray}}}which is equivalent expression for   \eqref{C1.opt1}.

Now problem \eqref{optbb1.eq} can be reformulated as a standard form for BB, which is given by 
{\setlength\abovedisplayskip{3pt} 
\setlength\belowdisplayskip{3pt}
{\small
\begin{subequations}\label{optbb2.eq}
\begin{align}
 \min_{\tilde{\beta},\Gamma}  &  -\sum_{m=1}^M\sum_{j_m=1}^{q_m}  \log_2\Big(1 + \Gamma_{j_m \rightarrow j_m}^m\Big) \\
\mathrm{s.t.} & \Gamma_{j_m \rightarrow j_m}^m \leq   \frac{g_{j_m}^m \beta_{j_m}^m }{ g_{j_m}^m \sum_{i_m=j_m+1}^{q_m} \beta_{i_m}^m +\sum_{n \neq m} g_{j_m}^n \beta^n + \sigma^2},\label{C1.optbb2}\\
& \eqref{C2.optbb1} \And  \eqref{C3.optbb1} \And  \eqref{C4.optbb1} \And \eqref{SICOrd2.eq}. \label{C2.optbb2}
\end{align}
\end{subequations}}}
\noindent
\begin{proposition}\label{Pro2}
Problem \eqref{optbb2.eq} is equivalent to problem \eqref{optbb1.eq}, hence the optimal solution to problem \eqref{optbb2.eq} is also optimal for problem \eqref{optbb1.eq}. \end{proposition}
 
 \begin{proof}
 With the introduce variable set $\{\Gamma_{j_m \rightarrow j_m}^m, ~j_m \in \mathcal{C}_m,~m\in \mathcal{M}\}$,   the objective in \eqref{optbb1.eq} can be expressed  as minimizing the following function
 {\setlength\abovedisplayskip{3pt} 
\setlength\belowdisplayskip{3pt}
{\small
 \begin{align}
-\sum_{m=1}^M\sum_{j_m=1}^{q_m}  \log_2\Big(1 + \Gamma_{j_m \rightarrow j_m}^m\Big).
 \end{align}}} with constraints \eqref{C1.optbb2} and \eqref{C2.optbb2}.
 We relax the equalities in \eqref{SINR.eq1} to be 
 {\setlength\abovedisplayskip{3pt} 
\setlength\belowdisplayskip{3pt}
{\small
\begin{align}\label{SINR.eq2}
\Gamma_{j_m\rightarrow j_m}^m \leq  \frac{g_{j_m}^m \beta_{j_m}^m }{ g_{j_m}^m \sum_{i_m=j_m+1}^{q_m} \beta_{i_m}^m +\sum_{n \neq m} g_{j_m}^n \beta^n + \sigma^2},
\end{align}}}\par Based on monotonic increasing feature of $\log(\cdot)$ function, \eqref{SINR.eq2} will be strict equality at optimum, which implies \eqref{optbb2.eq} and \eqref{optbb1.eq} have the same optimal solution.
 \end{proof}
 
  \vspace{-0.5cm} 
\subsection{Preliminaries for BB Algorithms} 
  \vspace{-0.3cm} 
In this subsection, we introduce the preliminary steps for BB algorithm. We start by transforming the constraint sets into a multi-dimensional box set. Then, we construct the bound function for each  multi-dimensional box set.  Finally, we propose a more effective algorithm to find the values of the bound functions.
 
 \subsubsection{Constructing box constraint sets}
 
 We first define the objective function in \eqref{optbb2.eq} and the feasible set for $\Gamma$ as $\mathcal{U}(\Gamma)$ and $\mathcal{G}$, respectively. 
 {\setlength\abovedisplayskip{3pt} 
\setlength\belowdisplayskip{3pt}
{\small
 \begin{align}
\mathcal{ U}(\Gamma) =& -\sum_{m=1}^M\sum_{j_m=1}^{q_m}  \log_2\Big(1 + \Gamma_{j_m \rightarrow j_m}^m\Big),\\
  \mathcal{G} =& \begin{Bmatrix}
 \Gamma | \eqref{C1.optbb2} \And \eqref{C2.optbb2}
\end{Bmatrix} .\label{FeSet.eq}
 \end{align}}}\par Note that it is true that the objective function $ \mathcal{U}(\Gamma) < 0$.  Therefore, the  optimization problem in \eqref{optbb2.eq} can be equivalently expressed as
 {\setlength\abovedisplayskip{3pt} 
\setlength\belowdisplayskip{3pt}
{\small
\begin{eqnarray}\label{opt5.eq}
\begin{aligned}
\min_{\Gamma} \quad &\mathcal{ U}(\Gamma)  \quad
\mathrm{s.t.}\quad & \Gamma  \in \mathcal{G}.
\end{aligned}
\end{eqnarray}}}\par Now the optimal objective value can be written as $p^\star = \inf_{\boldsymbol{\Gamma} \in \mathcal{G}} \mathcal{ U}(\boldsymbol{\Gamma})$. 
To formulate a standard form for  BB algorithm, let us define a new function as
{\setlength\abovedisplayskip{3pt} 
\setlength\belowdisplayskip{3pt}
{\small
\begin{align}\label{ExpdFun.eq}
\tilde{\mathcal{ U}}(\Gamma) = \begin{cases}
\mathcal{ U} (\Gamma) & \mathrm{if}~ \Gamma \in \mathcal{G} \\
0 & \mathrm{otherwise},
\end{cases}
\end{align}}}and note that for any set $\mathcal{S} \subseteq \mathbb{R}^{M_t}$,  we have 
{\setlength\abovedisplayskip{3pt} 
\setlength\belowdisplayskip{3pt}
{\small
\begin{align}\label{BoundFun.eq0}
\inf_{\Gamma \in \mathcal{S}} \tilde{\mathcal{ U}}(\Gamma) = \inf_{\Gamma \in \mathcal{G}}  \mathcal{ U} (\Gamma) = p^\star, 
\end{align}}}if $\mathcal{G} \subseteq \mathcal{S}$.
The first equality follows the fact that $ \mathcal{ U} (\Gamma)$  is a lower bound of $\tilde{\mathcal{ U}}(\Gamma)$ for $\Gamma \in  \mathcal{S}$.

Therefore, based on the feasible set in \eqref{FeSet.eq}, we can construct a $M_t$-dimensional rectangle $\mathcal{D}_0$ as
{\setlength\abovedisplayskip{3pt} 
\setlength\belowdisplayskip{3pt}
{\small
\begin{align}\label{D0.eq}
\mathcal{D}_0 = \begin{Bmatrix} 
\Gamma | \bar{\gamma}_{j_m}^m \leq \Gamma_{j_m \rightarrow j_m}^m \leq \overline{\Gamma}_{j_m \rightarrow j_m}^m, ~j_m \in \mathcal{C}_m,~m \in \mathcal{M}
\end{Bmatrix},
\end{align}}}which satisfies $\mathcal{G} \subseteq \mathcal{D}_0$.  
Here $\bar{\gamma}_{j_m}^m = 2^{\bar{R}_{j_m}}-1$ and $\overline{\Gamma}_{j_m \rightarrow j_m}^m$ is the upper bound of $\Gamma_{j_m \rightarrow j_m }^m$. It is easy to know that for each $\Gamma_{j_m \rightarrow j_m}^m$, it is upper bounded  by 
{\setlength\abovedisplayskip{3pt} 
\setlength\belowdisplayskip{3pt}
{\small
\begin{align}
\Gamma_{j_m \rightarrow j_m}^{m,m} \leq  \frac{ g_{j_m}^m P_{tot}}{ \sigma^2}.
\end{align}}}\par

 Note that for any $M_t$-dimensional rectangle $\mathcal{D} = \{\Gamma |   \underline{\Gamma}_{j_m \rightarrow j_m}^{m} \leq \Gamma_{j_m \rightarrow j_m}^{m} \leq \overline{\Gamma}_{j_m \rightarrow j_m}^m , ~ j_m \in \mathcal{C}_m,~m \in \mathcal{M}\} $ such that $\mathcal{D} \subseteq \mathcal{D}_0$. Based on the observation, we define a function $g(\boldsymbol{\Gamma}) $ as
 {\setlength\abovedisplayskip{3pt} 
\setlength\belowdisplayskip{3pt}
 {\small
\begin{align}\label{BoundFun.eq1}
g(\Gamma|\mathcal{D}) = \inf_{\Gamma \in \mathcal{D}} \tilde{\mathcal{ U}}(\Gamma).
\end{align}}}\par 
By combining \eqref{BoundFun.eq0} and \eqref{BoundFun.eq1}, one can obtain that 
{\setlength\abovedisplayskip{3pt} 
\setlength\belowdisplayskip{3pt}
 {\small
\begin{align}
g(\Gamma|\mathcal{D}_0) =  \inf_{\Gamma \in \mathcal{D}_0} \tilde{\mathcal{ U}}(\Gamma) = p^\star.
\end{align}}}\par

Through the above discussions, problem \eqref{optbb2.eq} has been converted into  a minimization of the non-convex function $\mathcal{ U}(\Gamma) $ over the a box constraint set $\mathcal{D}$. 
With BB algorithms, searching is organised by using a binary tree, where the initial box constraint set \eqref{FeSet.eq} will be subdivided iteratively into smaller subsets for searching. At each leaf node, we can obtain a lower bound and an upper bound for \eqref{optbb2.eq} by a bound function. Hence,  construction of the bound functions will be discussed in the following. 

\subsubsection{Construct  upper bound and lower bound function}

Based on the fact that $\tilde{\mathcal{ U}}$ in \eqref{ExpdFun.eq} is a non-decreasing function,  similar to \cite{balakrishnan1991branch,Weeraddana11TSP}, the lower bound function $\underline{g}$ and the  upper bound function $\overline{g}$ can be constructed as
{\setlength\abovedisplayskip{3pt} 
\setlength\belowdisplayskip{3pt}
 {\small
\begin{align}
\underline{g}(\Gamma) = \begin{cases} 
\mathcal{ U}(\overline{\Gamma}), & \underline{\Gamma} \in \mathcal{G} \\
0, & \mathrm{otherwise},
\end{cases}
\end{align}}}and 
{\setlength\abovedisplayskip{3pt} 
\setlength\belowdisplayskip{3pt}
 {\small
\begin{align}
\overline{g}(\Gamma) = \begin{cases} 
\mathcal{ U}(\underline{\Gamma}), & \underline{\Gamma} \in \mathcal{G} \\
0, & \mathrm{otherwise},
\end{cases}
\end{align}}}\par
Note that  to calculate $\underline{g}(\Gamma)$ and  $\overline{g}(\Gamma)$, the key step if to check if the condition $ \underline{\Gamma} \in \mathcal{G}$ is guaranteed.

Let $\{ \underline{\Gamma}_{j_m \rightarrow j_m}^m \}$ be a specified set of SINR values. Testing if these values are achievable is equivalent to solving the following feasibility problem
{\setlength\abovedisplayskip{3pt} 
\setlength\belowdisplayskip{3pt}
 {\small
\begin{eqnarray}\label{optbb3.eq}
\begin{aligned}
\mathrm{Find}  &\quad  \tilde{\beta}  \quad
\mathrm{s.t.}  &\quad \underline{\Gamma} \in \mathcal{G}.
\end{aligned}
\end{eqnarray}}}\par
Though problem \eqref{optbb3.eq} is a convex problem and can be solved directly, to further improve the computional efficiency, we develop a more efficient algorithm to check  $\underline{\Gamma} \in \mathcal{G}$ by exploiting the features of problem \eqref{optbb3.eq}.

\subsubsection{Solution for problem \eqref{optbb3.eq}}
We first consider constraint \eqref{C1.optbb2} for all $k\in \mathcal{K}$. Let $\boldsymbol{\Gamma} = [\Gamma_{1_1 \rightarrow 1_1}^1,\cdots,\Gamma_{q_1 \rightarrow q_1}^1,\cdots, \Gamma_{1_M \rightarrow 1_M}^M,\cdots,\Gamma_{q_M \rightarrow q_M}^M]^T$ with $\boldsymbol{\Gamma} \in \mathbb{R}^{M_t \times 1}$ and $\boldsymbol{\beta} = [\beta_{1}^1, \cdots,\beta_{q_1}^1,\cdots,\beta_{1_M}^M, \cdots,\beta_{q_M}^M]^T$ with $\boldsymbol{\beta} \in \mathbb{R}^{M_t \times 1}$. By rearranging \eqref{C1.optbb2} as
{\setlength\abovedisplayskip{3pt} 
\setlength\belowdisplayskip{3pt}
 {\small
\begin{eqnarray}\label{C1re.optbb2}
\begin{aligned}
\beta_{j_m}^m  -& \Gamma_{j_m \rightarrow j_m}^m \sum_{i_m=j_m+1}^{q_m} \beta_{i_m}^m - \frac{\Gamma_{j_m \rightarrow j_m}^m}{g_{j_m}^m}\sum_{n \neq m} g_{j_m}^n \beta^n  \geq \frac{\Gamma_{j_m \rightarrow j_m}^m}{g_{j_m}^m} \sigma^2,
\end{aligned} 
\end{eqnarray}}}\par 
Based on the transformation, \eqref{C1.optbb2} can be expressed as a compact form:
{\setlength\abovedisplayskip{3pt} 
\setlength\belowdisplayskip{3pt}
 {\small
\begin{align}\label{PerrFro.eq}
\big( \mathbf{I}_{M_t} - (\boldsymbol{\Lambda} + \mathbf{D}\mathbf{G})\big) \boldsymbol{\beta}  \succeq \sigma^2  \mathbf{D} \boldsymbol{1}_{M_t},
\end{align}}}where $\succeq$ or $\succ$ denotes the componentwise inequality between real matrix and vectors and 
{\setlength\abovedisplayskip{3pt} 
\setlength\belowdisplayskip{3pt}
 {\small
\begin{eqnarray*}\label{PerrFroMatrix.eq}
\begin{aligned}
&\boldsymbol{\Lambda} = \mathrm{diag}\left[\boldsymbol{\Lambda}_1, \cdots, \boldsymbol{\Lambda}_M\right],\\
&\mathbf{D}= \mathrm{diag}\begin{bmatrix}\frac{\Gamma_{1_1\rightarrow 1_1}^1}{g_{1_1}^1}, \cdots,  \frac{\Gamma_{q_1 \rightarrow q_1}^1}{g_{q_1}^1},\cdots, \frac{\Gamma_{1_M\rightarrow 1_M}^M}{g_{1_M}^M}, \cdots,  \frac{\Gamma_{q_M \rightarrow q_{M}}^M}{g_{q_{M}}^M}\end{bmatrix},\\
&\mathbf{G} = \begin{bmatrix} \begin{matrix}
\underbrace{\mathbf{G}_{-1},\cdots,\mathbf{G}_{-1}}\\q_1
\end{matrix}, \cdots,\begin{matrix}
\underbrace{\mathbf{G}_{-M},\cdots,\mathbf{G}_{-M}}\\q_M\end{matrix}
\end{bmatrix}^T,\\
&\mathbf{G}_{-m} = \begin{bmatrix}g_{1_m}^1\mathbf{1}_{q_1}^T, \cdots,g_{1_m}^{m-1}\mathbf{1}_{q_{m-1}}^T, \mathbf{0}_{q_m}^T,\cdots,g_{1_m}^{M}\mathbf{1}_{q_{M}}^T
\end{bmatrix},
\end{aligned}
\end{eqnarray*}}}where $\boldsymbol{\Lambda}_m$ is a upper triangular matrix with the $(j_m,k_m)$-th entry is $\Gamma_{j_m \rightarrow j_m}^m$ for $k_m>j_m$ and $\mathbf{G} \in _{M_t \times M_t}$.

\begin{lemma}\label{Lem1.irre}
 Let $\boldsymbol{\Lambda}$, $\mathbf{D}$ and $\mathbf{G}$ be given in \eqref{PerrFroMatrix.eq},  the following is satisfied:
 {\setlength\abovedisplayskip{3pt} 
\setlength\belowdisplayskip{3pt}
 {\small
\begin{eqnarray}
\boldsymbol{\Lambda}+ \mathbf{D}\mathbf{G} \succeq \mathbf{0},
\end{eqnarray}}}which implies that $\boldsymbol{\Lambda}+ \mathbf{D}\mathbf{G}$ is an irreducible nonnegative matrix.   
\end{lemma}
\begin{proof}
Note that  $\boldsymbol{\Lambda}$, $\mathbf{D}$ and $\mathbf{G}$ are nonnegative matrices and $\boldsymbol{\Lambda}$ is a diagonal matrix with positive entries. Thus, $\boldsymbol{\Lambda}+ \mathbf{D}\mathbf{G}$ is irreducible nonnegative  \cite{horn2012matrix} and Lemma \ref{Lem1.irre} is proved. 
\end{proof}

Based on Lemma \ref{Lem1.irre}, the following theorem helps us to check if $\underline{\Gamma} \in \mathcal{G}$, where  $\rho(\boldsymbol{\Lambda} + \mathbf{D}\mathbf{G})$ denotes the Perron-Frobenius eigenvalue of matrix $\boldsymbol{\Lambda} + \mathbf{D}\mathbf{G}$. 
\begin{theorem}\label{Theo1}
 When problem \eqref{optbb0.eq} is feasible, for any $\underline{\Gamma} \geq \bar{\gamma}$, $\bar{\gamma} = \{\bar{\gamma}_{j_m}^m, \forall j, \forall m\}$, we check if $\underline{\Gamma} \in \mathcal{G}$ by the following conditions:
 \begin{itemize}
 \item[1)] $\rho(\boldsymbol{\Lambda} + \mathbf{D}\mathbf{G}) \geq 1 \Rightarrow \underline{\Gamma} \not\in \mathcal{G}$;
 
 \item[2)] $\rho(\boldsymbol{\Lambda} + \mathbf{D}\mathbf{G}) < 1 \Rightarrow \boldsymbol{\beta} = \big(\mathbf{I}_{M_t} - (\boldsymbol{\Lambda} + \mathbf{D}\mathbf{G}) \big)^{-1}\sigma^2  \mathbf{D} \boldsymbol{1}_{M_t}$. If $\sum_{m=1}^M \sum_{j_m=1}^{q_m} \beta_{j_m}^m > P_{tot}$, $\underline{\Gamma} \not\in \mathcal{G}$;
 
 \item[3)] When $\rho(\boldsymbol{\Lambda} + \mathbf{D}\mathbf{G}) < 1$ and $\sum_{m=1}^M \sum_{j_m=1}^{q_m} \beta_{j_m}^m \leq P_{tot}$,  if    $\boldsymbol{\beta}$ satisfies the constraints in \eqref{C1.optbb1} for all $k,j$ and $m$, $\boldsymbol{\beta}$ is the corresponding optimal solution; otherwise, the corresponding optimal power allocation coefficients can be obtained by solving problem \eqref{optbb3.eq}.
 \end{itemize} 
\end{theorem} 
\begin{proof}
See Appendix B.
\end{proof}
Based on {\bf Theorem \ref{Theo1}}, we conclude the procedure of checking $\underline{\Gamma} \in \mathcal{G}$ in {\bf Algorithm \ref{alg:Check}}. 
\begin{algorithm}
\caption{Checking the condition for  $\underline{\Gamma} \in \mathcal{G}$}
\label{alg:Check}
\begin{algorithmic}[1]
\STATE {Construct matrices $\mathbf{\Lambda}$, $\mathbf{D}$ and $\mathbf{G}$ as in \eqref{PerrFroMatrix.eq}}
\STATE{ If $\rho(\boldsymbol{\Lambda} + \mathbf{D}\mathbf{G}) \geq 1$, $\underline{\Gamma} \not\in \mathcal{G}$ and STOP}
\STATE{ If $\rho(\boldsymbol{\Lambda} + \mathbf{D}\mathbf{G}) < 1$,  $\boldsymbol{\beta} = \big(\mathbf{I}_{M_t} - (\boldsymbol{\Lambda} + \mathbf{D}\mathbf{G}) \big)^{-1}\sigma^2  \mathbf{D} \boldsymbol{1}_{M_t}$.}
\IF{ $\sum_{m=1}^M \sum_{j_m=1}^{q_m} \beta_{j_m}^m > P_{tot}$}
\STATE{ $\underline{\Gamma} \not\in \mathcal{G}$ and STOP.}
\ELSE
\STATE{Solve problem \eqref{optbb3.eq} using standard convex tool. }
\STATE{ If \eqref{optbb3.eq} is feasible, then $\underline{\Gamma} \in \mathcal{G}$; Otherwise, $\underline{\Gamma} \not\in \mathcal{G}$. STOP.}
\ENDIF
\end{algorithmic}
\end{algorithm}

 \vspace{-0.9cm} 
\subsection{Proposed Optimal  User Scheduling and Power Allocation Algorithms}
 \vspace{-0.3cm} 


Based on the above discussions, the procedures of the proposed BB algorithm for optimal power allocation is described as follows. Let $\mathcal{D}_t = \{\mathcal{A}_1^1(t),\cdots,\cdots\mathcal{A}_{q_M}^M(t) \}$  denote the set of box subsets $\mathcal{A}_{j_m}^m(t)=\{ \underline{\Gamma}_{j_m \rightarrow j_m}^m(t) \leq\Gamma_{j_m \rightarrow j_m}^m \leq \overline{\Gamma}_{j_m \rightarrow j_m}^m(t) \}$ for all $j_m$ and $m$ at the $t$-th iteration. $\mathcal{D}(0)$ is the initial rectangular constraint set, which is defined in \eqref{D0.eq} on the root node of the binary tree. At the $t$-th iteration, we spilt $\mathcal{D}_t$ into two subsets  $\mathcal{B}_I$ and $\mathcal{B}_{II}$ along one of its longest edges, removing $\mathcal{D}(t)$ and adding the two new subsets to $\mathcal{R}(t)$. Next, we solve \eqref{optbb3.eq} based on {\bf Algorithm  \ref{alg:Check}} over each subset $\mathcal{B}_l, ~l\in \{ I, II\}$. A lower bound and a upper bound  can be obtained. Then, we choose the minimum over all upper bounds as $U(t)$ and choose the minimum over all lower bounds as $L(t)$,  i.e., taking the minimum over all the upper and lower bounds at each leaf node across all the levels in the binary tree. Removing the leaf node $\mathcal{D}$ such that $\underline{g}( \mathcal{D} ) \geq U(t)$, which will not affect the optimality of the BB tree. Repeat the above procedures until it satisfies the accuracy $\epsilon$ which is the difference between the global upper bound and the global lower bound. In the procedure of generating the BB tree, a sequence of subsets will be generated from $\mathcal{D}(0)$. The details are given in {\bf Algorithm \ref{alg:BB}} that captures the global optimal solution of \eqref{opt1.eq}.

\begin{algorithm}
\caption{The optimal power allocation algorithm based on BB}
\label{alg:BB}
\begin{algorithmic}[1]
\STATE { Initialization for BB:

1) Compute $\mathcal{D}(0)$ where   $\hat{\Gamma}_{\pi(j)\rightarrow\pi(j)}^m = \frac{ g_{k_m}^m P_{tot}}{ \sigma^2}$.      

2) Compute $U(1) = \overline{g} \left( \mathcal{D}_0 \right)$, $L(1) = \underline{g} \left( \mathcal{D}_0 \right)$ by solving problem \eqref{optbb3.eq}. 
 
3) Set $\{ \mathcal{R}(1) = \mathcal{D}_0 \}$, optimal lower bound $U^* = U(1)$, tolerance $\epsilon>0$ and $t=1$.}
\WHILE{$U(t)-L(t)>\epsilon$}
\STATE{ Pick $\mathcal{D} \in \mathcal{R}(t)$ for which $\underline{g}(\mathcal{D})=L(t)$ and set $\mathcal{D}(t) = \mathcal{D}$.}

\STATE{Subdivide $\mathcal{D}(t)$ along one of its longest edges into $\mathcal{B}_I$ and $\mathcal{B}_{II}$.
}

\STATE{Compute $\overline{g} \big( \mathcal{B}_I \big)$, $\underline{g} \big( \mathcal{B}_{II} \big)$ by solving problem \eqref{optbb3.eq}.}

\STATE{Update the upper bound $U(t)$ and the lower bound $L(t)$ as follows:

~~~~$L(t) = \underset{\mathcal{D} \in \mathcal{R}(t+1)}{\min}{\underline{g}({\mathcal{D}}})$;
		
~~~~$U(t) = \underset{\mathcal{D} \in \mathcal{R}(t+1)}{\min{\overline{g}}( \mathcal{D} )}$  ;

~~~~ update $U^* = \min(U^*,U(t))$. } \label{Step6}

\STATE{Update $\mathcal{R}(t+1)$ by removing all $\mathcal{B}_t$ for which $\underline{g}( \mathcal{D} ) \geq U(t+1)$. 

}
\STATE{$t:=t+1$.}
\ENDWHILE
\STATE{Output the absolute value  $|U^*|$ and the optimal power allocation $\beta$.}
\end{algorithmic}
\end{algorithm}
\begin{remark}
To ensure the global optimality, an exhaustive procedure is required. For ease of implementation, we select the bisection method to implement the subdivision of $\mathcal{D}(t)$ \cite{Horst1991}.
\end{remark}
For the set $\mathcal{D}(t)$,  let $v = \frac{1}{2}\big( \Gamma^i + \Gamma^j\big)$ denote the midpoint of the longest edge of the set $\mathcal{D}(t)$ and  $\Gamma^i$ and $\Gamma^j$ correspond to the vertexes of the longest of the edge.   Specifically, 
its subdivisions $\mathcal{B}_I$ and $\mathcal{B}_{II}$ produced by bisection can be obtained by replacing $\Gamma^i$ and $\Gamma^j$  by $v$ in $\mathcal{B}_I$ and $\mathcal{B}_{II}$.

\begin{proposition}\label{Pro3.BBConvergence}
 {\bf Algorithm \ref{alg:BB}} converges to the global optimal solution of problem \eqref{C1.optbb1}.
\end{proposition}
\begin{proof}
The convergence and the optimality can be proved by the following conditions:
\begin{itemize}
\item[1)] According to the characteristics of BB, at each iteration $t$, the function $g(\Gamma|\mathcal{D}(t))$ is bounded by the lower and upper bound functions: $\underline{g}(\Gamma|\mathcal{D}(t))  \leq g(\Gamma|\mathcal{D}(t)) \leq  \overline{g}(\Gamma|\mathcal{D}(t))$. 

\item[2)] The subdivision procedure is exhaustive since $\lim_{t \rightarrow \infty} V(\mathcal{D}(t)) = 0$, where $V(\mathcal{D}(t))$ denotes the size of $\mathcal{D}(t)$. Hence, the sequence of the global upper bound $U^*$ obtained by any infinite subdivisions with bisection is exhaustive.

\item[3)] By step \ref{Step6}, the minimization operations are performed on the lower and upper bounds. Hence, the global upper bound $ U^*(t+1) \leq U^*(t)$, which is a decreasing sequence.
\end{itemize}
Based on the above facts, {\bf Algorithm \ref{alg:BB}} searches every possible points in the feasible set $\mathcal{D}$ and thus is a global solution according to \cite{Horst1991}.
\end{proof}

\begin{remark}
At the $t$-th iteration of {\bf Algorithm \ref{alg:BB}}, $U(t)$ and $L(t)$ are the minimums over all the upper bounds and lower bounds at each leaf nodes in the BB tree, respectively,  which give a global upper bound and lower bound on the optimal value of \eqref{optbb1.eq}. The stopping criterion for {\bf Algorithm \ref{alg:BB}} can be $U(t)-L(t) \leq \epsilon$ for given a small $\epsilon$, which means that $  U^* - \epsilon \leq \mathcal{U}^{\mathrm{opt}}$.
\end{remark}

\begin{remark}
The overall complexity of {\bf Algorithm \ref{alg:BB}}  is determined by the complexity of at each iteration and the number of iterations required for achieving the desired tolerance. In general, the worst case computational complexity of {\bf Algorithm \ref{alg:BB}} is exponential in the number of variables. 
\end{remark}

In summary, the proposed joint  user scheduling and power allocation algorithm based on BB technique is summarised in {\bf Algorithm \ref{alg:JUSPA}}. In {\bf Algorithm \ref{alg:JUSPA}}, let $\Theta$ be the  all possible user scheduling combinations with all possible decoding order. For each search, the optimal power allocation is attained by BB algorithm.
\begin{algorithm}
\caption{Joint user scheduling and power allocation algorithm}
\label{alg:JUSPA}
\begin{algorithmic}[1]
\STATE{Construct the set $\Theta$ contains all possible user  scheduling  combinations and all possible decoding order. Set $n = 1$.}
\WHILE{$\Theta $ is not empty set}
\STATE{Check the feasibility of the given $c$ and $\pi$ by solving problem \eqref{optbb0.eq}.}
\STATE {Using  {\bf Algorithm \ref{alg:BB}} to solve problem \eqref{optbb1.eq}.}
\ENDWHILE
\STATE{$\mathcal{U}(n) = |U^*|$ and $n := n +1$. }
\STATE{Output the optimal objective value $\mathcal{U}^* = \max(\mathcal{U})$.}
\end{algorithmic}
\end{algorithm}

\begin{remark}
The complexity of {\bf Algorithm \ref{alg:JUSPA}} is determined by the search space of $\Theta$ and the complexity at each search. As known the exhaust search is exponential complexity with $\mathcal{O}(K^{M_t})$.
\end{remark}

 \vspace{-0.5cm} 
\section{Low Complexity Solutions}
 \vspace{-0.3cm} 
The computation is cumbersome  to the global solution, specially when the size of the problem becomes large. In order to reduce the computational complexity, our goal in this section is to propose a low complexity  algorithm that obtains a suboptimal solution of problem \eqref{opt1.eq}.  

 \vspace{-0.9cm} 
\subsection{SCA-based Suboptimal Power Allocation Algorithms}
 \vspace{-0.3cm} 
To begin with, we  consider the power allocation in \eqref{optbb1.eq} for the given the scheduled users and decoding order.  In this subsection, we develop a low complexity power allocation algorithm based on first-order approximations and SCA. 

To handle the nonconvex objective function in \eqref{optbb1.eq}, we approximate the the nonconvex objective by the following lower bound \cite{Papan06ICC}:
{\setlength\abovedisplayskip{3pt} 
\setlength\belowdisplayskip{3pt}
 {\small
\begin{align}\label{loglb.eq}
\mu\ln(\tau) + \nu \leq \ln(1+\tau), 
\end{align}}}where
{\setlength\abovedisplayskip{3pt} 
\setlength\belowdisplayskip{3pt}
 {\small
\begin{align}\label{loglbPar.eq}
\begin{cases} \mu &= \frac{\tilde{\tau}}{1+\tilde{\tau}}\\
\nu &= \log(1+\tilde{\tau})-\frac{\tilde{z}}{1+\tilde{\tau}}
 \end{cases}
\end{align}}}\par
The approximation in \eqref{loglb.eq} is tight at a chosen value $\tilde{\tau}$ when the constants $\{\mu,\nu \}$ are chosen as \eqref{loglbPar.eq}. Thus, given a set of fixed $\{\mu,\nu \}$,  problem  \eqref{optbb1.eq} can be approximated as follows.
{\setlength\abovedisplayskip{3pt} 
\setlength\belowdisplayskip{3pt}
 {\small
\begin{subequations}\label{optsca1.eq}
\begin{align}
\max_{\beta} \quad &  \sum_{m=1}^M\sum_{j_m=1}^{q_m}\frac{1}{\ln_2} \Big(\mu_{j_m}^m \ln(\mathrm{SINR}_{j_m \rightarrow j_m}^m) + \nu_{j_m}^m \Big) \\
\mathrm{s.t.}\quad  &\sum_{n\neq m} \Big(g_{k_m}^m  g_{j_m}^n - g_{j_m}^m g_{k_m}^n\Big)\beta^n + \big(g_{k_m}^m - g_{j_m}^m \big) \sigma^2 \geq 0, \label{C1.optsca1}\\
& \sum_{m=1}^M \sum_{j_m=1}^{q_m} \beta_{j_m}^m \leq P_{tot},\label{C2.optsca1}\\
&\mu_{j_m}^m \ln(\mathrm{SINR}_{j_m \rightarrow k_m}^m) + \nu_{j_m}^m \geq \ln_2\bar{R}_{j_m}, \label{C3.optsca1}\\
&k_m>j_m,~j_m, k_m \in \mathcal{C}_m,~m \in \mathcal{M}.\label{C4.optsca1}
\end{align}
\end{subequations}}}\par 
However, \eqref{optsca1.eq} is still non-convex, since the objective function and constraint \eqref{C3.optsca1} is not concave in $\beta$.  To proceed further, a variable transformation $x_{j_m}^m = \ln(\beta_{j_m}^m)$ is introduced.  As a result, for any $j_m \in \mathcal{C}_m$ and $k_m \in \mathcal{C}_m$ with $k_m \geq j_m$, we have 
{\setlength\abovedisplayskip{3pt} 
\setlength\belowdisplayskip{3pt}
 {\small
\begin{eqnarray}
\begin{aligned}
& \ln(\mathrm{SINR}_{j_m \rightarrow k_m}^m) =\ln(g_{k_m}^m) + x_{j_m}^m -  \ln\Big( \sum_{i_m=j_m+1}^{q_m} g_{k_m}^m e^{x_{i_m}^m} + \sum_{n \neq m}g_{k_m}^n e^{x^n} + \sigma^2 \Big),
\end{aligned}
\end{eqnarray}}}\par
Now, we consider the constraint in \eqref{C1.optsca1} becomes
{\setlength\abovedisplayskip{3pt} 
\setlength\belowdisplayskip{3pt}
 {\small
\begin{align} \label{C1re.optsca1}
\sum_{n\neq m} g_{j_m}^m g_{k_m}^ne^{x^n}- \big(g_{k_m}^m - g_{j_m}^m \big) \sigma^2 \leq \sum_{n\neq m} g_{k_m}^m  g_{j_m}^n e^{x^n},
\end{align}}}which is non-convex. However it can be approximated by applying the first-order Taylor approximation when giving a point $\tilde{x}^n$. Let $F(x^n) = \sum_{n\neq m} g_{k_m}^m  g_{j_m}^n e^{x^n}$.
{\setlength\abovedisplayskip{3pt} 
\setlength\belowdisplayskip{3pt}
 {\small
\begin{eqnarray}\label{Taylor.eq}
\begin{aligned}
F(x^n) &=  F(\tilde{x}^n) + \nabla_{x^n} F(\tilde{x}^n)(x^n - \tilde{x}^n) \\
&= F(\tilde{x}^n) + \sum_{n\neq m} g_{k_m}^m  g_{j_m}^n \sum_{i_n=1}^{q_n} (x_{i_n}^n - \tilde{x}_{i_n}^n).
\end{aligned}
\end{eqnarray}}}\par
Substituting \eqref{C1re.optsca1} and  \eqref{Taylor.eq}  into problem \eqref{optsca1.eq}, we can obtain the following approximation of problem \eqref{optsca1.eq}:
{\setlength\abovedisplayskip{3pt} 
\setlength\belowdisplayskip{3pt}
 {\small
\begin{eqnarray}\label{optsca2.eq}
\begin{aligned}
\max_{\beta} \quad &   \sum_{m=1}^M\sum_{j_m=1}^{q_m}\frac{1}{\ln_2} \Big(\mu_{j_m}^m \ln(\mathrm{SINR}_{j_m \rightarrow j_m}^m) + \nu_{j_m}^m \Big) \\
\mathrm{s.t.}\quad  & g_{j_m}^m g_{k_m}^n\beta^n - \big(g_{k_m}^m - g_{j_m}^m \big) \sigma^2  \leq F(\tilde{x}^n) + \\&\sum_{n\neq m} g_{k_m}^m  g_{j_m}^n \sum_{i_n=1}^{q_n} (x_{i_n}^n - \tilde{x}_{i_n}^n), \\
& \sum_{m=1}^M \sum_{j_m=1}^{q_m} \beta_{j_m}^m \leq P_{tot},\label{C2.optsca2}\\
&\mu_{j_m}^m \ln(\mathrm{SINR}_{j_m \rightarrow j_m}^m) + \nu_{j_m}^m \geq \ln_2\bar{R}_{j_m}, \label{C3.optsca2}\\
&k_m>j_m,~j_m, k_m \in \mathcal{C}_m,~m \in \mathcal{M},\label{C4.optsca2}
\end{aligned}
\end{eqnarray}}}\par
Problem \eqref{optsca2.eq} is a convex optimization problem. It can be solved by a standard convex tool such as cvx \cite{cvx}. 
\begin{remark}
Problem \eqref{optsca2.eq} is a lower bound approximation of problem \eqref{optbb1.eq} because of the relaxation in \eqref{loglb.eq} and the first-order approximation in \eqref{Taylor.eq}. 
\end{remark}

Since problem \eqref{optsca2.eq} is obtained by approximating problem \eqref{optbb1.eq} at a feasible point set $\{\tilde{x}_{i_n}^n\}$. The approximation can be further improved by successively approximating problem \eqref{optbb1.eq} based on the optimal solution $\{\tilde{x}_{i_n}^n\}$ obtained  by solving problem  \eqref{optsca2.eq} in the previous approximation. Therefore, the proposed successive approximation approach can be described in the following.
\begin{algorithm}
\caption{SCA algorithm for solving \eqref{optbb1.eq}}
\label{alg:SCA}
\begin{algorithmic}[1]
\STATE{Initialize  a set of feasible power allocation coefficient $\beta$. }
\STATE{Compute the objective value in \eqref{optsca2.eq}, denoted as $\Phi[0]$. Set $t=1$.}
\WHILE{$\frac{|\Phi[t] - \Phi[t-1]|}{\Phi[t-1]} \leq \epsilon'$, where $\epsilon'$ is a given stopping criterion.}
\STATE{t := t+1.}
\STATE{Solve problem \eqref{optsca2.eq}, to obtain the optimal solution $\{\phi_{j}^m[t] ,~j\in \mathcal{K}\}$ and $\boldsymbol{\beta}[t]$.}
\ENDWHILE
\STATE{Output the optimal $\{\phi_{j}^m[t] ,~j\in \mathcal{C}\}$ and $\boldsymbol{\beta}$ }
\end{algorithmic}
\end{algorithm}

\begin{remark}
In each iteration of {\bf Algorithm \ref{alg:SCA}}, the sum rate will be improved successively. However, due to the total power constraint, the generated sum rate sequence is bounded, which implies the convergence of {\bf Algorithm \ref{alg:SCA}}.
\end{remark}

\begin{remark}
Since the approximation in \eqref{loglb.eq} and  \eqref{Taylor.eq} are lower bound approximation for problem \eqref{optbb1.eq}, the solution generated by  {\bf Algorithm \ref{alg:SCA}} is suboptimal.
\end{remark}

 \vspace{-0.5cm} 
\subsection{Many-to-One Matching Algorithm for User Scheduling}
 \vspace{-0.3cm} 
To avoid combinatorial complexity in exhausting search, in this section, we propose a low complexity user scheduling algorithm based on matching theory \cite{Gale62,bando2012many}.
Given the user power allocation coefficients, the optimization problem \eqref{opt1.eq} can be transformed into
{\setlength\abovedisplayskip{3pt} 
\setlength\belowdisplayskip{3pt}
 {\small
 \begin{subequations}\label{opt1match.eq}
\begin{align}
 \max_{\mathbf{c}} \quad &  \mathcal{H} = \sum_{m=1}^M\sum_{j=1}^{q_m} R_{\pi_m(j)\rightarrow\pi_m(j)}^m\\
\mathrm{s.t.}\quad  &\eqref{C4.opt1}-\eqref{C6.opt1}
\end{align}
\end{subequations}}}which can be formulated as  a many to one bipartite matching problem with externalities
among users \cite{bando2012many}. Based on the concept of stable matching, we will develop a low complexity matching algorithm in the following.

\subsubsection{Preliminaries of matching theory in user scheduling}
Based on the definitions of  $\mathcal{M}$ and $\mathcal{K}$ in Section \ref{SigMod}, one can know that $\mathcal{M}$ and $\mathcal{K}$ are disjoint sets. In NOMA, each beam can support multiple users simultaneously, but each user is allowed to access for at most one beam. 
Thus, in matching, there exists a positive quota $q_m$ which indicates the number of users a beam has to support. The quota for each beam may be different. This problem is to match the users to the beams. This is a many-to one matching problem. These types of problems have a long history in economics, such as the marriage problem ($q_m=1$)\cite{Gale62} and  workers/firms problem \cite{bando2012many} or hospitals/residents problem\cite{sasaki1996two} with $q_m>1$.

\begin{definition}
A many-to-one matching  $\varphi$ is a function from the set $\mathcal{M} \bigcup \mathcal{K}$ into the set of unordered families of elements of $\mathcal{M} \bigcup \mathcal{K}\bigcup \{0\}$ such that
\begin{itemize}
\item[1)] $|\varphi(k)| \leq 1$ for every user $k \in \mathcal{K}$;

\item[2)]   $|\varphi(m)| = q_m$ for every beam $m \in \mathcal{M}$;

\item[3)] $\varphi(k) \in \mathcal{M}$ if and only if $k \in \varphi(M)$;

\item[4)] $k \in \varphi(m) \Leftrightarrow \varphi(k)=m$.

\end{itemize}
\end{definition} \label{Def1}
The notation $\varphi$ has different meanings depending on the parameter. If the parameter is a user $k$, then $\varphi(k)$ maps to the matched beam. If the parameter is a beam $m$, then $\varphi(m)$ gives the set of matched users.

\begin{proposition}
The user paring problem can be formulated as a many-to-one matching problem with  externalities among  users.  
\end{proposition}
\begin{proof}
From Definition \ref{Def1}, one can easy obtain that  the user paring problem in \eqref{opt1match.eq} is a many to one matching problem. Due to the inter-beam interference and the intra-beam interference existed for each user's achievable  rate,  each beam' preferences depend not only on users whom it support, but also on users whom the other beams support. Similarly, each user's preferences is not related with the only  beam it matched but all of the beams.  Based on these features, one can conclude that this problem is a  a many-to-one matching problem with with externalities \cite{bando2012many,sasaki1996two,bodine2011peer}.
\end{proof}

To model the  externalities, we define the preference value for the user $k$ on beam $m$ as  the achievable  rate:
{\setlength\abovedisplayskip{3pt} 
\setlength\belowdisplayskip{3pt}
 {\small
\begin{align}
\mathcal{H} _{k}^m = \log_2 \Big( 1 + \Gamma_{k}^m \Big).
\end{align}}}\par
Then we define the preference value of beam $m$ as
{\setlength\abovedisplayskip{3pt} 
\setlength\belowdisplayskip{3pt}
 {\small
\begin{align} \label{Rate_m.eq}
\mathcal{H} ^m = \sum_{k \in \varphi(m)}\log_2 \Big( 1 + \Gamma_{k}^m \Big).
\end{align}}}\par
Thus, in this matching model, each beam $m$ has a strict preference ordering $\succ_m$ over $\mathcal{K}$. Each user also has a preference relation $\succ_k$ over the set $\mathcal{M}\bigcup \{0\}$, where $\{0\}$ denotes the user is unmatched.  Specifically, for a given user $k$, any two beam $m$ and $m'$ with $m,m' \in \mathcal{M}$, any two matchings $\varphi$ and $\varphi'$ is defined as
{\setlength\abovedisplayskip{3pt} 
\setlength\belowdisplayskip{3pt}
 {\small
\begin{align}
(m,\varphi) \succ_k (m',\varphi') \Leftrightarrow  \mathcal{H}_{k}^m (\varphi) >  \mathcal{H}_{k}^m (\varphi'),
\end{align}}}which indicates that user $k$ prefers beam $m$ in $\varphi$ to beam $m'$ in $\varphi'$ only if user $k$ can achieve a higher rate on beam $m$ than beam $m'$. Similarly, for any beam $m$, its preference $\succ_m$ over the user set can be describe as follows. For any two subsets of users $C$ and $C'$ with $C \neq C'$ and any two matchings  $\varphi$ and $\varphi'$, $C = \varphi(m)$, $C' = \varphi'(m)$:
{\setlength\abovedisplayskip{3pt} 
\setlength\belowdisplayskip{3pt}
 {\small
\begin{align}
(C,\varphi)\succ_m (C',\varphi') \Leftrightarrow  \mathcal{H}^m (\varphi) >  \mathcal{H}^m (\varphi'),
\end{align}}}which denotes that beam $m$ prefers the set of users $C$ to $C'$ only when beam $m$ can get a higher rate from $C$.

Since externalities exist in the formulated matching problem, it is not straightforward to define a stability concept because a stablility of a matching depends on how a deviating pair expects the reaction of the other agents \cite{bando2012many}. To tackle the externalities, the two-sided exchange stability 
 has been introduce in \cite{bodine2011peer}. Based on the concept of two-sided exchange-stable matchings, we propose a matching algorithm for the user paring problem in the next subsection.

\subsubsection{Designs of many-to-one matching algorithm}

To define exchange stability, it is convenient to first define a swap matching $\varphi_k^j$ in which user $k$ and user $j$ switch beams while keeping other users' assignments the same. We define a swap operation among the users to exchange their matched beams. A swap matching between user $j$ and user $k$ is define as follows. 
\begin{definition}
A swap matching is defined as $\varphi_k^j = \{ \varphi \setminus \{(k,m),(j,n)\}  \bigcup \{(j,m),(k,n)\} \}$
\end{definition} 
where $\varphi(k)=m$ and $\varphi(j) = n$.  To approve a swap operation, we introduce the concept of swap-blocking pair.
\begin{definition}
Given a matching function $\varphi$ and a pair of users $(k,j)$, if there exist $m=\varphi(k)$ and $n = \varphi(j)$ such that 
\begin{itemize}
\item[1)] $\forall x \in \{k,j,m,n \}$, $U_m(\varphi_k^j) \geq U_m(\varphi)$;

\item[2)] $\exists x \in \{k,j,m,n \}$, such that $U_m(\varphi_k^j) > U_m(\varphi)$,
\end{itemize}  
then the swap matching $\varphi_{k}^j$ is approved, and $(k,j)$ is called a swap-blocking pair in $\varphi$.
\end{definition}
The features of the swap-blocking pair ensure that if a swap matching is approved, then the achievable rates of any user involved will not decrease, and at least one user's achievable rate will increase.

Based on the above discussions, we can describe the users' behaviours in the many-to-one matching with externalities as bellow. Every two users can be arranged by the BS to form a candidate swap blocking pair. The BS checks whether they can benefit each other by exchanging their matches without hurting the interests of corresponding beams.   Through a series of  swap operations, the  matching can reach a stable status, also known as a two-sided exchange stable
matching defined as below.
\begin{definition}
A matching $\varphi$ is two-sided exchange-stable (2ES) if and only if there does not exist a swap-blocking pair.
\end{definition}

With the definition of above, we conclude the proposed user scheduling algorithm in {\bf Algorithm \ref{alg:C}}.

\begin{algorithm}
\caption{ User scheduling based on Matching Theory }
\label{alg:C}
\begin{algorithmic}[1]
\STATE{ Initialize the candidate user set $\mathcal{A}$ by {\bf Algorithm \ref{alg:GS}}. }
\REPEAT
\STATE{ For any user $k \in \mathcal{A}$, it searches for another user $j \in \mathcal{A}\setminus \mathcal{A}(\varphi(k))$. 
 }
\IF{$k,j$ is a swap-blocking pair}
\STATE{$\varphi = \varphi_k^j$}
\ELSE
\STATE{Keep the current matching state}
\ENDIF
\UNTIL{No swap-blocking pair is found}
\STATE{The stable matching $\varphi$}
\end{algorithmic}
\end{algorithm}

\begin{algorithm}
\caption{Initialization Algorithm}
\label{alg:GS}
\begin{algorithmic}[1]
\REQUIRE ~~\\
\STATE{Initialize the preference lists for all users and beams based on the scalar channel gain $|\mathbf{h}_k \mathbf{w}_{m'}|^2,~k \in \mathcal{K}~ \mathrm{and}~m' \in \mathcal{M}$. }
\STATE{ Set the user set of accepted by beam $m$  $\mathcal{A}^0(m) = \emptyset$ ,  the set of rejected users $\mathcal{W}^{0}(m) = \emptyset$ and the set of rejected beams $\mathcal{W}^{0}(k) =\emptyset$. Set $t=0$.}
\REPEAT
\STATE{$t := t+1$}
\STATE{All users not yet assigned $k \in \mathcal{K} \backslash \bigcup_{m\in \mathcal{M}} \mathcal{A}^{t-1}$ propose to their current best beam that has not reject user $k$, i.e.
$m = \mathrm{arg} \max_{m' \in \mathcal{M}\setminus \mathcal{W}(k)^{t-1} } |\mathbf{h}_k \mathbf{w}_{m'}|^2$.   }
\STATE{Denote the users who propose to beam $m$ as $\tilde{k}^m_1,\cdots,\tilde{k}^m_{s'}$.}
\STATE{Beam $m$ accepts the first $q_m$ best ranked users from $\mathcal{S} = \{ k_1^m,\cdots,k_s^m\} = \mathcal{A}^{t-1}(m) \bigcup \{\tilde{k}^m_1,\cdots,\tilde{k}^m_{s'}\}$ and update $\mathcal{A}^t(m) = \{k_1^m,\cdots ,k_{q_m}^m\}$,  where $s$ is the total number propose to the beam $m$. }.
\STATE{Update the set of rejected users $\mathcal{W}^{t-1}(m) = \{ k_{q_m+1}^m,\cdots,k_{s}^m\}$ and the set of rejected beams $\mathcal{W}^{t-1}(k) = \{m \in \mathcal{M}: k \in \mathcal{W}^{t-1}(m)\}$. 
 }
\UNTIL{All beams are achieved its maximum number of users or each remained user has been rejected by all beams.}
\STATE{Output  $\varphi$ and $\mathcal{A} = \{\mathcal{A}^{t}(m), ~m=1,\cdots,M\}$.  }
\end{algorithmic}
\end{algorithm}
\begin{remark}
Note that the initialization algorithm is a deferred acceptance algorithm \cite{Gale62}, the complexity mainly lies in the number of the user proposing. In the worst case, the proposing number is $KM$. In addition, the maximum number of swap operations in {\bf Algorithm \ref{alg:C}} is $M^2q_m^2$. 
\end{remark}

 \vspace{-0.5cm} 
\section{Simultion results}
 \vspace{-0.3cm} 
In this section, simulations are conducted to evaluate the proposed algorithms. We consider the channel model described in \eqref{Chmodl.eq}, with a number of paths $L = 3$. The AoDs are assumed to take continuous values, and are uniformly distributed in $[0,2\pi]$. The BS randomly generated $M$ orthogonal beam. The mmWave system is assumed to operated at $28$ GHz carrier frequency. The bandwidth of the system is assumed $100$ MHz  and with path-loss exponent $c_{LoS} = 2$ and $c_{NLoS} = 3$.   In the following simulations, we assume that the users are uniformly distributed in a single cell with radius $R_c$ and  the  SNR in the plots is defined as $\mathrm{SNR} = \frac{P \eta}{\sigma^2M}$ \cite{Alkhateeb15TWC}. In addition, the stopping criteria $\epsilon = 0.1$ and $\epsilon' = 0.05$. In addition, we assume $q_1=\cdots=q_M=q$, which indicates that each beam can be occupied by $q$ users simultaneously. All  users have the same QoS constraint is they are scheduled, i.e., $\bar{R}_j = R_{th}, ~j \in \mathcal{C} $.

We first evaluate convergence of the proposed  BB algorithm and the SCA algorithm  solving problem \eqref{optbb1.eq} in Fig. \ref{Convergence.Fig} for different SNR.  As can be observed from Fig. \ref{Convergence.Fig}, both of the proposed BB  and SCA algorithms are converged for different SNR. In the proposed BB algorithm, 
the upper bound and the lower bound  become tighter as the number of  iterations grows. In addition, though some performance loss has been caused by the proposed SCA algorithm,  the convergence speed of SCA is much faster than the proposed BB algorithm. The reason is that  the BB algorithm performs a  bisection division process  for each dimension, which approaches  to the  exhaustive search in a small scale of $\epsilon$. 

In Fig. \ref{ComVsAlg.Fig}, we investigate the sum rate versus the SNR both in mmWave NOMA systems and mmWave OMA systems with  different algorithms. 
As can be observed from Fig.  \ref{ComVsAlg.Fig}, the sum rate of all algorithms increases monotonically with the SNR. This is because the sum rate can be improved by optimizing  user scheduling and power allocation via solving the problem in \eqref{opt1.eq}. However, the multiuser mmWave system is interference limited due to the inter-beam interference exist, the sum rate will not be improved with increasing the SNR.  In particular,  three different algorithms solving problem \eqref{opt1.eq} are plotted in Fig. \ref{ComVsAlg.Fig}: the global optimal algorithm-Exhaust+BB, the moderate complexity algorithm- Matching+BB and the low complexity algorithm - Matching+SCA. As shown, the sum rate of Exhaust+BB  
grows faster than Matching+BB and Matching+SCA at the cost of the high complexity.  Besides, compared with  Exhaust+BB and Matching+BB,    Matching+SCA  achieves a good sum rate performance. Particularly,  a same sum rate can be obtaibed with Matching+BB in the SNR regions of $0 \sim 10$ dB which indicates that the proposed suboptimal power allocation  algorithm is efficient to solving problem \eqref{optbb1.eq}. 
In addition, it can be observed from Fig. \ref{ComVsAlg.Fig}, the sum rate of the mmWave NOMA system outperforms that in the conventional mmWave OMA systems. This reveals that the application of NOMA into mmWave can further improve the spectral efficiency.

Due to the high free-space path loss, there are different $\beta$ values on different mmWave frequencies. We examine the effects on the sum rate with different mmWave frequencies in Fig. \ref{ComVsFreq.Fig}.
 Fig. \ref{ComVsFreq.Fig} illustrates the sum rate versus the SNR at different mmWave frequency  for $K=100$ for $f_c=28$ GHz and $f_c = 60$ GHz. 
We observe that the proposed mmWave NOMA system can achieve high sum rate under $f_c = 28$ GHz compared to $f_c = 60$ GHz, due to the fact that mmWave link at 60 GHz has higher LoS and NLoS path loss exponents than that at 28 GHz, which leads to lower signal strength at users. In addition, the gap between mmWave NOMA and  mmWave OMA decreases when the SNR becomes large. The reason is that at the high SNR regions, the multiuser mmWave system  becomes  interference-limited. In this case, the inter-beam interference becomes the main factor  to restrain  the sum rate increases in the mmWave NOMA system the  mmWave OMA system .   
The impact of the proposed power allocation algorithm using SCA  on the sum rate under different mmWave frequencies is also plotted in Fig. \ref{ComVsFreq.Fig}. To validate the effectiveness, we compare the proposed SCA algorithm with the fixed power allocation scheme in mmWave-NOMA, called Match+Fixed NOMA. In Match+Fixed NOMA, we assume that the total power is distributed uniformly on  each beam, and the power allocation between the users in each beam is assumed to be $\beta_1$ and $\beta_2$ for the users with the better equivalent channel gain and the poorer equivalent channel gain, respectively. As can be observed that the proposed SCA algorithm can enhance the sum rate efficiently compared to the fixed power allocation scheme.

In Fig. \ref{ComVsNumUsers.Fig}, we investigate the sum rate of the mmWave system versus the total number of users for SNR = 10 dB, $q =2$. $R_{th} = 0.02$. Here, the average number of the selected users are fixed with $Mq$ for different algorithms. As can be observed from Fig. \ref{ComVsNumUsers.Fig},  the sum rate increases with the total number of users for the curves of Maching+SCA and Matching+Fixed Power. The reason is that  the inter-beam interference can be suppressed greatly  when the number of users  increases by the proposed matching algorithm.  However,the performance of  the random user scheduling algorithm is unsatisfied this is because the inter-beam interference can not be suppressed via random user scheduling. In addition, some users with very poor channel conditions will be scheduled which will decrease the total sum rate. Therefore,  user scheduling is  important for the proposed  mmWave NOMA systems. Furthermore, note that from Fig. \ref{ComVsNumUsers.Fig}, this increasing trend becomes slower as the total number of users becomes larger, since when the total number of users becomes large enough, the inter-beam interference will approach to constant. 

	In Fig. \ref{ComVsqm.Fig}, we investigate the total sum rate versus maximum numbers of users sharing the same beam, $q$ in Matching+SCA. Different total number of users are considered with $K=100$ and $K=200$. 
	 As can be observed from Fig. \ref{ComVsqm.Fig}, the proposed mmWave NOMA system can achieve the better sum rate when the SNR increases. 	Besides, compared to the case of $K=100$, the sum rate can be improved by increasing the number of users.  
It can  also be observed that the sum rate increases with increasing $q$, because more users are accessed to the same resource. Hence, the mmWave system is capble of obtain more performance gains by applying NOMA. Furthermore, compare with the gap between $q=1$ and $q=2$, the gap becomes smaller from $q=2$ to $q=3$ when $K=100$, which is because of   the total transmission power constraint at the BS.

\begin{figure}[!htb] 
\begin{tabular}{cc}
\begin{minipage}[t]{0.48\linewidth}
\centering
\includegraphics[width= 1\linewidth]{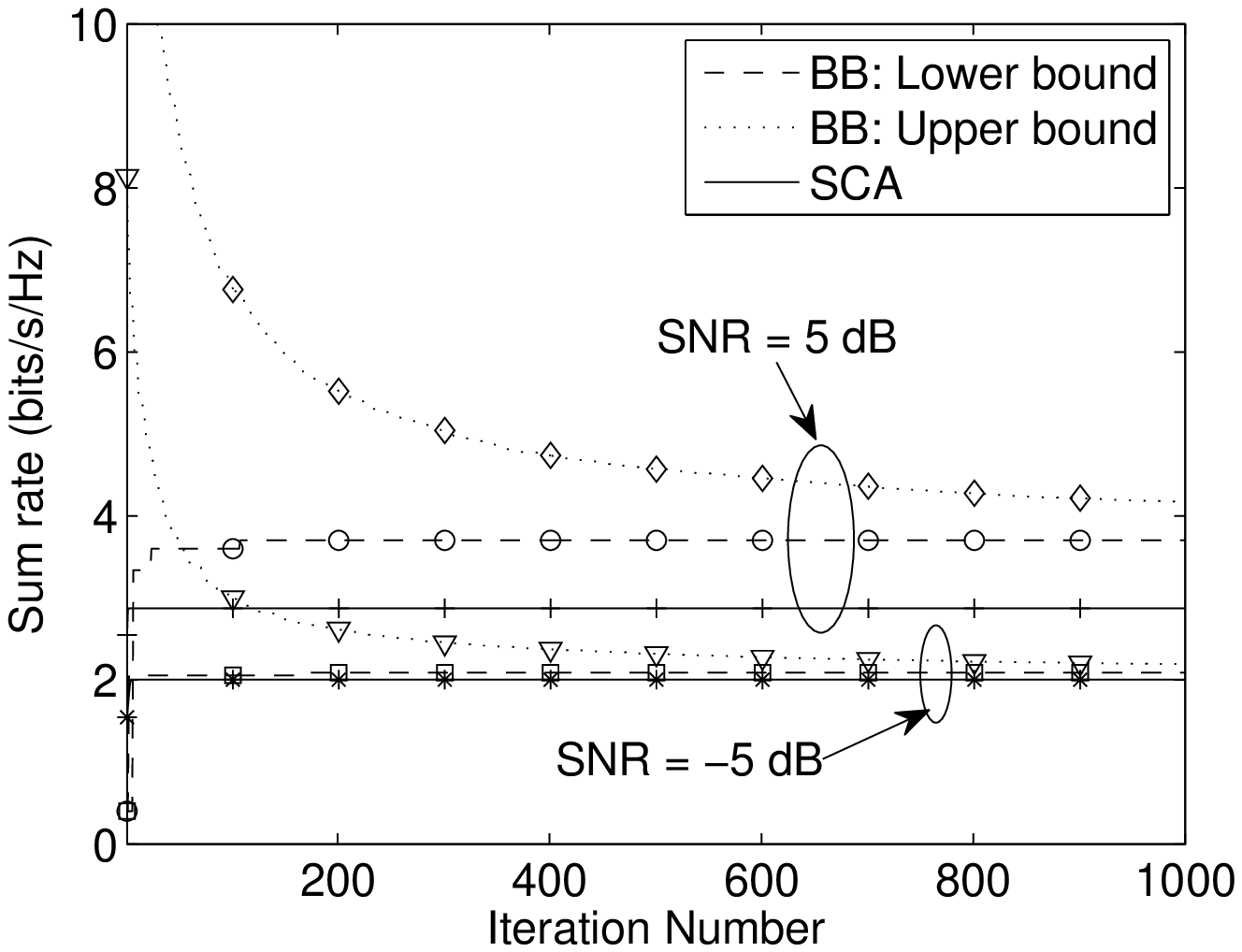}
 \vspace*{-1.5em} \caption{Comparisons of sum rate over iteration numbers between BB and SCA:   $N_{\mathrm{RF}} = 4$, $M = 2$, $q = 2$, $R_{th} = 0.1~ \mathrm{bits/s/Hz}$ and $R_c = 10$m.}\label{Convergence.Fig}
\end{minipage}
\begin{minipage}[t]{0.48\linewidth}
\centering
\includegraphics[width= 1\linewidth]{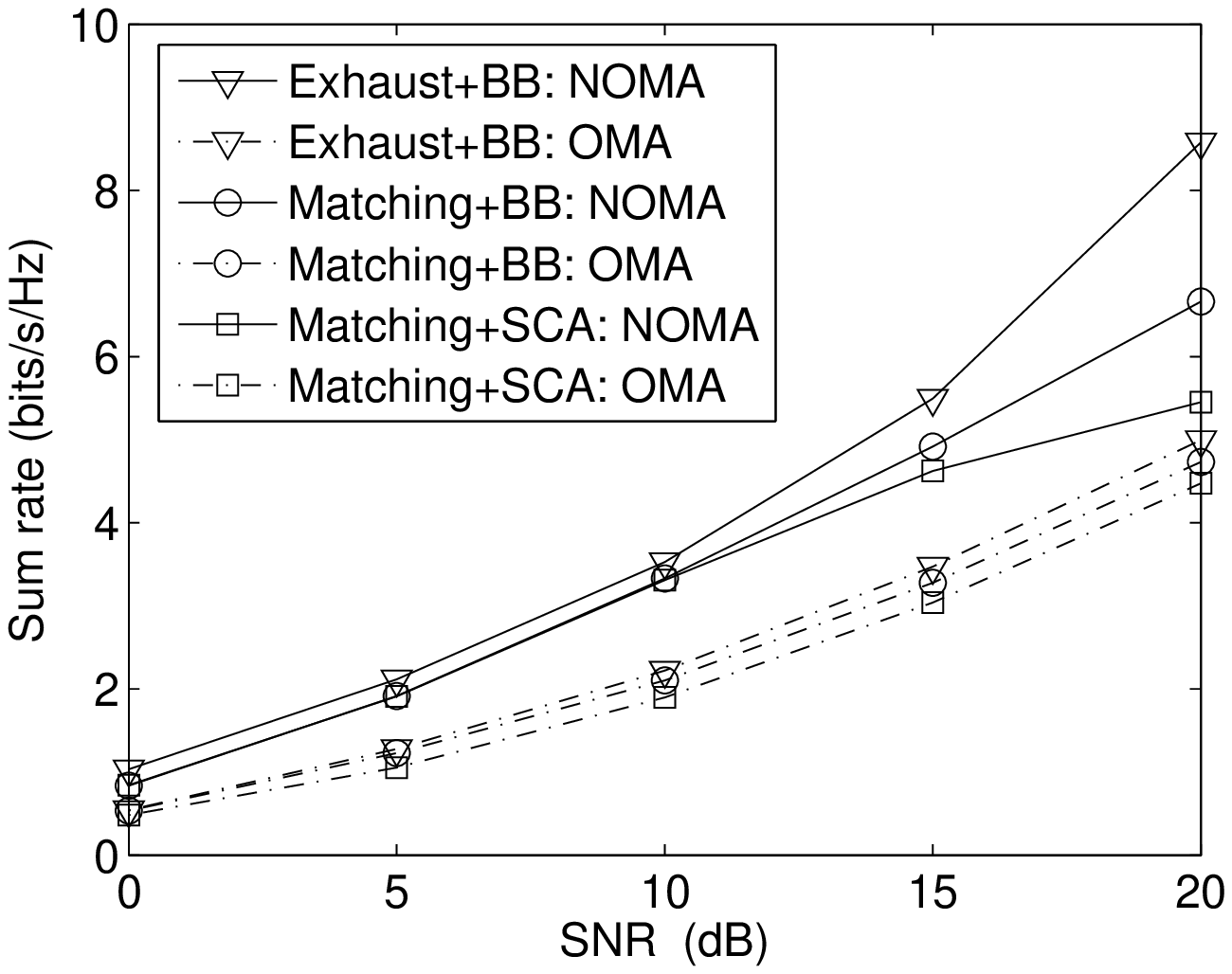}
 \vspace*{-1.5em} \caption{Comparisons of sum rate over different algorithms: $R_{th} = 0.1~ \mathrm{bits/s/Hz}$, and $R_c = 10$m.}\label{ComVsAlg.Fig}
 \end{minipage}
 \end{tabular}
\end{figure}
 \vspace{-1em}

\begin{figure} [t!]
\centering
\includegraphics[width= 3.8in, height=2.6in]{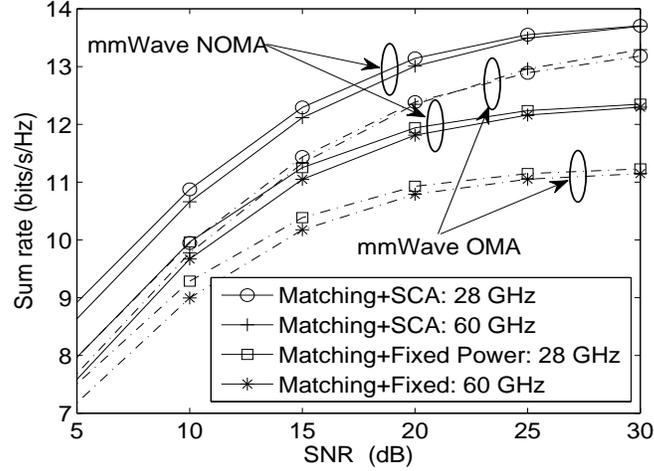}
 \vspace*{-1.5em} \caption{Comparisons of sum rate over different SNR at 28 GHz and 60 GHz: $R_{th} = 0.1~ \mathrm{bits/s/Hz}, ~R_c = 10$m. The LoS and NLoS path exponents are set based on the practical channel measurements \cite{Deng15ICCW,Rappaport12ICC}: $c_{\mathrm{LoS}} = 2$, $c_{\mathrm{NLoS}} = 3$ on $f_c=28$ GHz and $c_{\mathrm{LoS}} = 2.25$, $c_{\mathrm{NLoS}} = 3.71$ on $f_c=60$ GHz.}\label{ComVsFreq.Fig}
  \vspace{-1.5em}
\end{figure}

\begin{figure}[!htb] 
\begin{tabular}{cc}
\begin{minipage}[t]{0.48\linewidth}
\centering
\includegraphics[width= 1\linewidth]{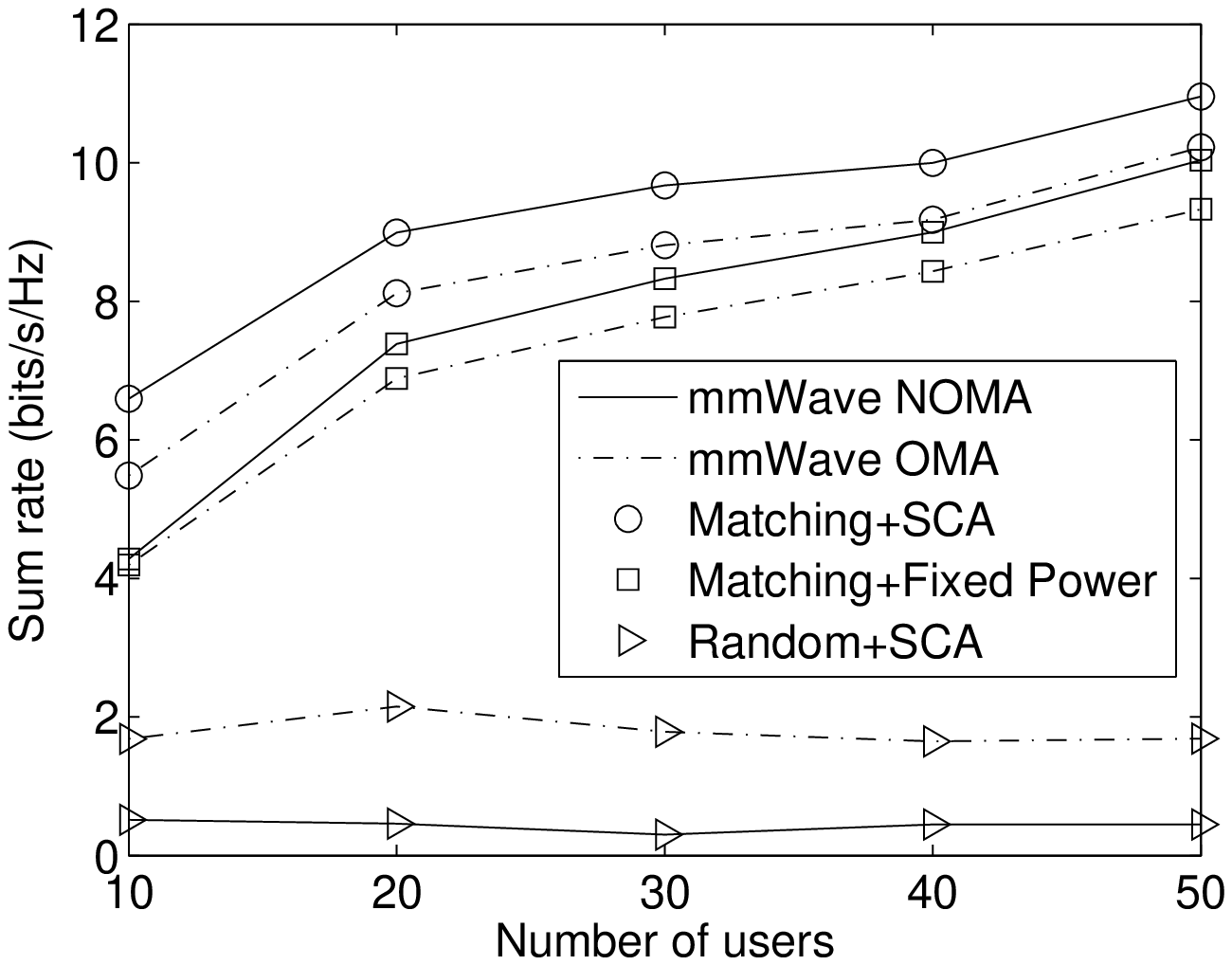}
 \vspace*{-1.5em} \caption{Comparisons of sum rate over SNR with different number of users: $R_{th} = 0.02~ \mathrm{bits/s/Hz},~R_c = 5$m and $\beta_1= \frac{1}{4}$ and $\beta_2 = \frac{3}{4}$.}\label{ComVsNumUsers.Fig}
 \end{minipage}
\begin{minipage}[t]{0.48\linewidth}
\centering
\includegraphics[width= 1\linewidth]{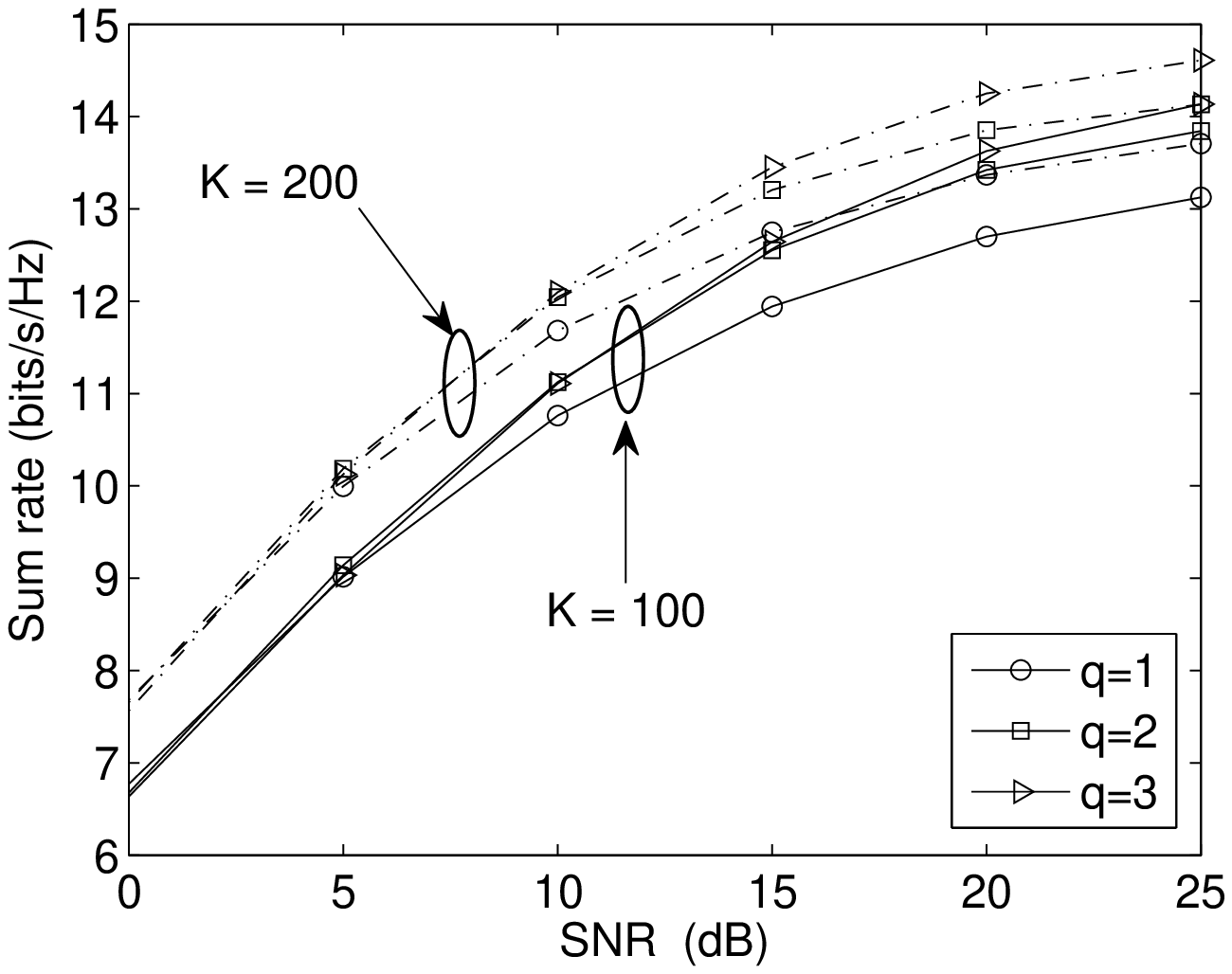}
 \vspace*{-1.5em} \caption{Comparisons of sum rate over SNR with different $q_m$: $R_{th} = 0.02~ \mathrm{bits/s/Hz},
  ~R_c = 10$m.}\label{ComVsqm.Fig}
  \end{minipage}
  \end{tabular}
   \vspace{-1.5em}
\end{figure}

  \vspace{-0.5cm} 
\section{Conclusions}
 \vspace{-0.3cm} 

In this paper, the designs of user scheduling and power allocation algorithms for mmWave NOMA systems with random beamforming were considered. Particularly, the formulated problem for the maximization of the sum rate of the mmWave NOMA system was a mixed integer programming. The original problem have been  into two subproblems and solved independently: 1) for the integer optimization of the user scheduling, exhaust search is adopted for a small scale problem; 2)  BB was applied for solving the  power allocation problem optimally. The generated optimal user scheduling and power allocation solution was served as a benchmark due to its prohibitive computational complexity. Moreover, a low complexity suboptimal algorithm was  developed to strike a trade-off between the performance and complexity, where user scheduling scheme and power allocation scheme  were designed based on matching theory and  SCA approach, respectively. Simulation results  have been showed that the proposed suboptimal algorithm   achieved a near optimal performance with low complexity compared to the global algorithm. In addition, our results showed that  the  sum rate of mmWave NOMA systems outperformed the conventional mmWave OMA systems.

 \vspace{-0.7em}
\numberwithin{equation}{section}
\section*{Appendix~A: Proof of Theorem \ref{Pro1}} \label{Appdx1}
\renewcommand{\theequation}{A.\arabic{equation}}
\setcounter{equation}{0}

Base on the computational complexity theory, to show the problem \eqref{opt1.eq} is NP-hard, we follows the following three steps: 1) choose a suitable known NP-complete decision problem $\mathcal{Q}$; 2) construct a polynomial time transformation from any instance of $\mathcal{Q}$ to an instance of problem \eqref{opt1.eq}; 3) prove the two instances have the same objective value under the transformation. 
In this paper,  to prove  problem \eqref{opt1.eq} is NP-hard, we divide the proof into two steps: $q_m = 1$ and $q_m > 1$.
\begin{itemize}
\item[1)] 
 We first consider the case $q_m=1$, \eqref{opt1.eq} becomes a joint power and user scheduling problem in the conventional OMA systems. The sum rate maximization problem in \eqref{opt1.eq} becomes the following  form: 
 {\setlength\abovedisplayskip{3pt} 
\setlength\belowdisplayskip{3pt}
 {\small
\begin{eqnarray}\label{opt1Simp1.eq}
\begin{aligned}
 \max_{\beta, c} \quad &  \sum_{m=1}^M R_{j \rightarrow j}^m \\
\mathrm{s.t.}\quad   & \sum_{m=1}^M  \beta_{j}^m \leq P_{tot},~j\in \mathcal{K},~m \in \mathcal{M},
\end{aligned}
\end{eqnarray}}}which has been proved to be NP-hard in \cite{Luo08JSAC}. 

\item[2)] When $q_m>1$, we prove that \eqref{opt1.eq}  is NP-hard even known the power allocation. In the following, we will construct an instance of problem \eqref{opt1.eq} with known power allocation coefficients.  First, the three-dimentional matching is known to be NP-hard.  We then consider an instance with $q_m=2$. Assuming that the users are  equally divided into two disjoined sets $\mathcal{K}_1$ and $\mathcal{K}_2$ satisfying the size $|\mathcal{K}_1|=|\mathcal{K}_2|=\frac{K}{2}$, $\mathcal{K}_1 \bigcup \mathcal{K}_2 = \mathcal{K}$ and $\mathcal{K}_1 \bigcap \mathcal{K}_2 = \emptyset$. In addition, we assume that the two  users $j$ and $k$ on beam $m$ are selected such that  $j \in \mathcal{K}_1$ and $K \in \mathcal{K}_2$, respectively.  Let $V$ be a subset of $\mathcal{M} \times \mathcal{K}_1 \times \mathcal{K}_2$, where the element $V_l = (m_l,k^1_l,k^2_l) \in V$. According to \eqref{Rate_m.eq}, the sum rate of any triple $V_l$ can be denoted as $\mathcal{H}_{V_l}$. Next, we need to determine if there exist a set $V' \subseteq V$ with the size $|V'|=\min \{M,\frac{K}{2}\}$ such that $\sum_{l=1}^{|V'|} \mathcal{H}_{V'_l} \leq \lambda$, where any $V'_l \in V'$ and $V'_n \in V'$ do not contain the same elements. Based on the definition, $V' \subseteq V$  will be a  three-dimentional matching when the following conditions hold: 1) $|V'|=\min \{M,\frac{K}{2}\}$; 2) For any two distinct triples: $(m_l,k^1_l,k^2_l) \in V$ and $(m'_l,k'^1_l,k'^2_l) \in V'$, we have $m_l\neq m'_l ,k^1_l \neq k'^1_l ,k^2_l \neq k'^2_l$. When $\lambda$ goes to nongative infinity,  problem \eqref{opt1.eq} with known power allocation coefficients becomes a three-dimensional matching problem. Therefore, the decision problem of the constructed instance is NP-complete and the corresponding instance is NP-hard.

Since a special case of problem \eqref{opt1.eq} is NP-hard, the original problem in \eqref{opt1.eq} is NP-hard.

\end{itemize} 
From the analysis of the above two cases, one can conclude that problem \eqref{opt1.eq} is NP-hard.

 \vspace{-0.9em}
\section*{Appendix~B: Proof of Theorem \ref{Theo1}} \label{Appdx2}
\renewcommand{\theequation}{B.\arabic{equation}}
\setcounter{equation}{0}

Theorem \ref{Theo1} is similar to the classical feasibility conditions in \cite{viswanath2003sum}. These conditions are derived based on Perron-Frobenius theory \cite{horn2012matrix} by assuming the primitiveness of $\boldsymbol{\Lambda}+ \mathbf{D}\mathbf{G}$. Different from  the conventional OMA systems, in which  only  the total transmission power constraint is considered, here we give a more general proof for NOMA system with the constraints of the deconding order.

To begin with, we show that $\rho(\boldsymbol{\Lambda} + \mathbf{D}\mathbf{G}) < 1$ is the necessary condition for $\underline{\Gamma} \in \mathcal{G}$.  Base on \eqref{PerrFro.eq}, we can construct the necessary condition for $\underline{\Gamma} \in \mathcal{G}$: if $\underline{\Gamma} \in \mathcal{G}$, then $ \exists \boldsymbol{\beta} \succeq \mathbf{0} $ such that 
{\setlength\abovedisplayskip{3pt} 
\setlength\belowdisplayskip{3pt}
 {\small
\begin{align}
\big( \mathbf{I}_{M_t} - (\boldsymbol{\Lambda} + \mathbf{D}\mathbf{G})\big) \boldsymbol{\beta} > \sigma^2  \mathbf{D} \boldsymbol{1}_{M_t}.
\end{align}}}ignoring the constraints in \eqref{C1.optbb1} and \eqref{C2.optbb1}.
Since each element of $\boldsymbol{\Gamma} $ satisfying  $\Gamma_{j_m\rightarrow j_m} \geq \bar{\gamma}_{j_m}$ is strict positive, which indicates that $\sigma^2  \mathbf{D} \boldsymbol{1}_{M_t} \succ 0$ and $\boldsymbol{\beta} \succ 0$. Based on these results, we can further refine the above necessary condition as follows:  if $\underline{\Gamma} \in \mathcal{G}$, then $ \exists \boldsymbol{\beta} \succeq \mathbf{0} $ such that 
{\setlength\abovedisplayskip{3pt} 
\setlength\belowdisplayskip{3pt}
 {\small
\begin{align}\label{PerrFro2.eq}
\big( \mathbf{I}_{M_t} - (\boldsymbol{\Lambda} + \mathbf{D}\mathbf{G})\big) \boldsymbol{\beta} \succ \boldsymbol{0},
\end{align}}}neglecting the constraints in \eqref{C1.optbb1} and \eqref{C2.optbb1}. Then based on the properties of the Perron-Frobenius eigenvalue stated in \cite{horn2012matrix}, a positive solution to $\boldsymbol{\beta}$ that satisfies \eqref{PerrFro2.eq} exists is and only if $\rho\big((\boldsymbol{\Lambda} + \mathbf{D}\mathbf{G})\big) <1$. Consequently, we the above necessary condition can be equivalently expressed as:  if $\underline{\Gamma} \in \mathcal{G}$, then $\rho\big((\boldsymbol{\Lambda} + \mathbf{D}\mathbf{G})\big) <1$. By contrast,  if  $\rho\big((\boldsymbol{\Lambda} + \mathbf{D}\mathbf{G})\big) \geq 1$, then $\underline{\Gamma} \notin \mathcal{G}$.

The second condition 2) follows from {\bf Proposition \ref{Pro2}}, where the SINR constraint in \eqref{C1.optbb2} are such that equalities, i.e., $\big( \mathbf{I}_{M_t} - (\boldsymbol{\Lambda} + \mathbf{D}\mathbf{G})\big) \boldsymbol{\beta} = \sigma^2  \mathbf{D} \boldsymbol{1}_{M_t}$. Moreover, $\rho(\boldsymbol{\Lambda} + \mathbf{D}\mathbf{G})\big) <1$, consequently, $\mathbf{I}_{M_t} - (\boldsymbol{\Lambda} + \mathbf{D}\mathbf{G})$ is invertible and its inverse has nonnegative entries, i.e., $\mathbf{I}_{M_t} - (\boldsymbol{\Lambda} + \mathbf{D}\mathbf{G})^{-1} \succeq \mathbf{0}$ \cite{horn2012matrix}. Thus, $\boldsymbol{\beta} = \big(\mathbf{I}_{M_t} - (\boldsymbol{\Lambda} + \mathbf{D}\mathbf{G}) \big)^{-1}\sigma^2  \mathbf{D} \boldsymbol{1}_{M_t} \succ \mathbf{0}$.

The second part of 2) is to showing that $\boldsymbol{\beta}^* = \big(\mathbf{I}_{M_t} - (\boldsymbol{\Lambda} + \mathbf{D}\mathbf{G}) \big)^{-1}\sigma^2  \mathbf{D} \boldsymbol{1}_{M_t}$ is the minimum power vector which sttisfies the SINR constraints in \eqref{C1.optbb2}. It is equivalently to verify that $\boldsymbol{\beta}^*$ is the optimal solution of the following  linear  vector  optimization  problem:
{\setlength\abovedisplayskip{3pt} 
\setlength\belowdisplayskip{3pt}
 {\small
\begin{eqnarray}\label{OptPower.eq}
\begin{aligned}
\min_{\boldsymbol{\beta}} \quad  \boldsymbol{\beta} \quad
\mathrm{s.t.} \quad  \big( \mathbf{I}_{M_t} - (\boldsymbol{\Lambda} + \mathbf{D}\mathbf{G})\big) \boldsymbol{\beta}  \succeq \sigma^2  \mathbf{D} \boldsymbol{1}_{M_t},
\end{aligned}
\end{eqnarray}}}which is convex \cite{boyd2004convex}. Hence the optimal solution satisfies thee KKT conditions, which are given as follows:
{\setlength\abovedisplayskip{3pt} 
\setlength\belowdisplayskip{3pt}
 {\small
\begin{subequations}
\begin{align}
& (\boldsymbol{\Lambda} + \mathbf{D}\mathbf{G})\boldsymbol{\lambda} = \mathbf{I}  , \label{eq1.kkt}\\
& \Big(\big( \mathbf{I}_{M_t} - (\boldsymbol{\Lambda} + \mathbf{D}\mathbf{G})\big) \boldsymbol{\beta} - \sigma^2  \mathbf{D} \boldsymbol{1}_{M_t}\Big)\boldsymbol{\lambda} = \mathbf{0},\label{eq2.kkt}\\
 &\big( \mathbf{I}_{M_t} - (\boldsymbol{\Lambda} + \mathbf{D}\mathbf{G})\big) \boldsymbol{\beta}  \succeq \sigma^2  \mathbf{D} \boldsymbol{1}_{M_t}.\label{eq3.kkt}
\end{align}
\end{subequations}}}where $\boldsymbol{\lambda} \succ \mathbf{0}$ is the  Lagrange multiplier vector. From \eqref{eq1.kkt} and \eqref{eq2.kkt}, one can obtain that $\boldsymbol{\lambda} = \mathbf{0}$. Then from \eqref{eq2.kkt} and \eqref{eq3.kkt}, it can be derived that $\big( \mathbf{I}_{M_t} - (\boldsymbol{\Lambda} + \mathbf{D}\mathbf{G})\big) \boldsymbol{\beta} = \sigma^2  \mathbf{D} \boldsymbol{1}_{M_t}$. Therefore, the optimal solution of \eqref{OptPower.eq} is give by $\boldsymbol{\beta}^* = \big(\mathbf{I}_{M_t} - (\boldsymbol{\Lambda} + \mathbf{D}\mathbf{G}) \big)^{-1}\sigma^2  \mathbf{D} \boldsymbol{1}_{M_t}$. As a result, if  $\sum_{m=1}^M \sum_{j_m=1}^{q_m} \beta_{j_m}^m > P_{tot}$, $\boldsymbol{\Gamma} \notin \mathcal{G}$.

Finally, we prove the condition 3) in {\bf Theorem \ref{Theo1}}. Since the constraints  of SIC decoding order, some solution attained by condition 3) may not satisfy the inequalities in \eqref{C1re.optbb2}. If so, the optimal power allocation can be obtained by solving \eqref{optbb2.eq} directly.


\vspace{-0.3cm}
{\small
 \bibliographystyle{IEEEtran}
  \linespread{1.1}\selectfont
\bibliography{ref[NOMA]mmWave}
}
\end{document}